\documentclass{article}

\usepackage{fullpage}
\usepackage{amsmath,amssymb,subfigure,graphicx,multirow}
\usepackage{amsthm}

\usepackage[ruled,lined,noend,linesnumbered]{algorithm2e}
\SetArgSty{textrm}
\usepackage{enumerate}
\usepackage{verbatim}
\usepackage{bm}
\usepackage{pseudocode}
\usepackage{tabularx,times}
\usepackage{algorithm2e}

\newcolumntype{Y}{>{\centering\arraybackslash}X}

\let \originalleft \left
\let\originalright\right
\renewcommand{\left}{\mathopen{}\mathclose\bgroup\originalleft}
\renewcommand{\right}{\aftergroup\egroup\originalright}

\newcommand{\Tl}{T_L}
\newcommand{\Tr}{T_R}
\newcommand{\bh}{\hat{h}}
\newcommand{\node}{\mbox{\tt Node}}
\newcommand{\leaf}{\mbox{\tt Leaf}}
\newcommand{\func}[1]{{\sc #1}}
\newcommand{\join}{\func{Join}}
\renewcommand{\split}{\func{Split}}
\newcommand{\joinTwo}{\func{Join2}}
\newcommand{\filter}{\func{Filter}}
\newcommand{\forallnew}{\func{ForAll}}
\newcommand{\union}{\func{Union}}
\newcommand{\insertnew}{\func{Insert}}
\newcommand{\delete}{\func{Delete}}
\newcommand{\intersect}{\func{Intersect}}

\newcommand{\build}{\func{Build}}
\newcommand{\expose}{\func{Expose}}
\newcommand{\difference}{\func{Difference}}

\newcommand{\range}{\func{Range}}
\newcommand{\textal}[1]{{{\sc {#1}}}}

\newcommand{\bound}{{$O\left(m\log (\frac{n}{m}+1)\right)$}}
\newcommand{\boundcontent}{{m\log \left(\frac{n}{m}+1\right)}}

\newcommand{\hide}[1]{} %hide

\newcommand{\para}[1]{\vspace{0.01in}\noindent\textbf{#1 }}
\newcommand{\betaone}{{\frac{1}{\beta}}}
\newcommand{\proofoutline}{{\textbf{Proof outline.}}}
\newcommand{\apply}[1]{{({\rm{applying~(#1)}})}}

\newtheorem{definition}{Definition}

\newtheorem{theorem}{Theorem}

\newtheorem{corollary}{Corollary}
\newtheorem{lemma}{Lemma}

\newcommand{\fullin}[1]{{\if 1<2
{#1}
\else
{}
\fi}}

\begin{document}

\title{Parallel Ordered Sets Using Join}

%\numberofauthors{4} %  in this sample file, there are a *total*
% of EIGHT authors. SIX appear on the 'first-page' (for formatting
% reasons) and the remaining two appear in the \additionalauthors section.
%
\author{
Guy Blelloch\\
Carnegie Mellon University\\
guyb@cs.cmu.edu
\and
Daniel Ferizovic\\
Karlsruhe Institute of Technology\\
dani93.f@gmail.com
\and
Yihan Sun\\
Carnegie Mellon University\\
yihans@cs.cmu.edu
}
%\author{
%\alignauthor Guy Blelloch\\
%       \affaddr{Carnegie Mellon University}\\
%       \email{guyb@cs.cmu.edu}
%\alignauthor
%Daniel Ferizovic\\
%       \affaddr{Karlsruhe Institute of Technology}\\
%       \email{dani93.f@gmail.com }
%\alignauthor
%Yihan Sun\\
%       \affaddr{Carnegie Mellon University}\\
%       \email{yihans@cs.cmu.edu}
%}
% There's nothing stopping you putting the seventh, eighth, etc.
% author on the opening page (as the 'third row') but we ask,
% for aesthetic reasons that you place these 'additional authors'
% in the \additional authors block, viz.

\maketitle
\begin{abstract}
  Ordered sets (and maps when data is associated with each key) are
  one of the most important and useful data types.  The set-set
  functions union, intersection and difference are particularly useful
  in certain applications.  Brown and Tarjan first described an
  algorithm for these functions, based on 2-3 trees, that meet the
  optimal $\Theta\left(\boundcontent\right)$ time
  bounds in the comparison model ($n$ and $m \le n$ are the input
  sizes).  Later Adams showed very elegant algorithms for the
  functions, and others, based on weight-balanced trees.  They only
  require a single function that is specific to the balancing
  scheme---a function that joins two balanced trees---and hence can be
  applied to other balancing schemes.  Furthermore the algorithms are
  naturally parallel.  However, in the twenty-four years since, no one
  has shown that the algorithms are work efficient (or optimal), sequential or
  parallel, and even for the original weight-balanced trees.

  In this paper we show that Adams' algorithms are both work efficient
  and highly parallel (polylog depth) across four different balancing
  schemes---AVL trees, red-black trees, weight balanced trees and
  treaps.  To do this we need careful, but simple, algorithms for
  \join{} that maintain certain invariants, and our proof is (mostly)
  generic across the schemes.

  To understand how the algorithms perform in practice we have also
  implemented them (all code except \join{} is generic across the
  balancing schemes).  Interestingly the implementations on all four
  balancing schemes and three set functions perform similarly in time
  and speedup (more than 45x on 64 cores).  We also compare the
  performance of our implementation to other existing libraries and
  algorithms including the standard template library (STL)
  implementation of red-black trees, the multicore standard template
  library (MCSTL), and a recent parallel implementation based on
  weight-balanced trees.  Our implementations are not as fast as the
  best of these on fully overlapping keys (but comparable), but better
  than all on keys with a skewed overlap (two Gaussians with different
  means).
\end{abstract}

\clearpage
\section{Introduction}
\label{intro}
Ordered sets and ordered maps (sets with data associated with each
key) are two of the most important data types used in modern
programming.  Most programming languages either have them built in as
basic types (e.g. python) or supply them as standard libraries (C++,
C\# Java, Scala, Haskell, ML).  These implementations are based on
some form of balanced tree (or tree-like) data structure and, at
minimum, support lookup, insertion, and deletion in logarithmic time.
Most also support set-set functions such as union, intersection, and
difference.  These functions are particularly useful when using
parallel machines since they can support parallel bulk updates.  In
this paper we are interested in simple and efficient parallel
algorithms for such set-set functions.

The lower bound for comparison-based algorithms for union,
intersection and difference for inputs of size $n$ and $m \le n$ and output an
ordered structure\footnote{By ``ordered structure'' we mean any data structure
that can output elements in sorted order without any comparisons.}
 is
$\log_2 {m+n \choose n} =
\Theta\left(\boundcontent\right)$.
%($m+n \choose n$ is the number of possible ways $n$ keys can be interleaved with $m$keys).
Brown and Tarjan first described a sequential algorithm for
merging that asymptotically match these bounds~\cite{brown1979fast}.
It can be adapted for union, intersection and difference with the same
bounds.  The bound is interesting since it shows that implementing
insertion with union, or deletion with difference, is asymptotically
efficient ($O(\log n)$ time), as is taking the union of two equal
sized sets ($O(n)$ time).  However, the Brown and Tarjan algorithm is
complicated, and completely sequential.

\begin{figure*}[th]
\hrulefill\\[-.1in]
\begin{minipage}[t]{.45\columnwidth}
\begin{lstlisting}[numbers=none]
@\textbf{insert}@$(T,k)$ =
  $(T_L,m,T_R)$ = split$(T,k)$;
  join$(\Tl,k,\Tr)$@\vspace{.1in}@
@\textbf{delete}@$(T,k)$ =
  $(\Tl,m,\Tr)$ = split$(T,k)$;
  join2$(\Tl,\Tr)$@\vspace{.1in}@
@\textbf{split}@$(T,k)$ =
  if $T$ = Leaf then (Leaf,false,Leaf)
  else $(L,m,R)$ = expose$(T)$;
    if $k = m$ then ($L$,true,$R$)
    else if $k < m$ then
      $(L_L,b,L_R)$ = split$(L,k)$;
      ($L_L$,$b$,join$(L_R,m,R)$)
    else $(R_L,b,R_R)$ = split$(R,k)$;
      (join$(L,m,R_L),b,R_R)$  @\vspace{.1in}@
@\textbf{splitLast}@($T$) =
  $(L,k,R)$ = expose$(T)$;
  if $R$ = Leaf then ($L,k$)
  else $(T',k')$ = splitLast($R$);
    (join($L,k,T'$),$k'$)@\vspace{.1in}@
@\textbf{join2}@($\Tl$,$\Tr$) =
  if $\Tl$ = $\leaf$ then $\Tr$
  else $(\Tl',k)$ = splitLast($\Tl$);
    join($\Tl',k,\Tr$)
\end{lstlisting}
\end{minipage}
\begin{minipage}[t]{\columnwidth}
\begin{lstlisting}[numbers=none]
@\textbf{union}@($T_1$,$T_2$) =
  if $T_1$ = Leaf then $T_2$
  else if $T_2$ = Leaf then $T_1$
  else ($L_2$,$k_2$,$R_2$) = expose($T_2$);
    ($L_1$,$b$,$R_1$) = split($T_1$,$k_2$);
    $\Tl$ = union($L_1$,$L_2$) || $\Tr$ = union($R_1$,$R_2$);
    join($\Tl$,$k_2$,$\Tr$) @\vspace{.1in}@
@\textbf{intersect}@($T_1$,$T_2$) =
  if $T_1$ = Leaf then Leaf
  else if $T_2$ = Leaf then Leaf
  else ($L_2$,$k_2$,$R_2$) = expose($T_2$);
    ($L_1$,$b$,$R_1$) = split($T_1$,$k_2$);
    $\Tl$ = intersect($L_1$,$L_2$) || $\Tr$ = intersect($R_1$,$R_2$);
    if $b$ = true then join($\Tl$,$k_2$,$\Tr$)
    else join2($\Tl$,$\Tr$) @\vspace{.1in}@
@\textbf{difference}@($T_1$,$T_2$) =
  if $T_1$ = Leaf then Leaf
  else if $T_2$ = Leaf then $T_1$
  else ($L_2$,$k_2$,$R_2$) = expose($T_2$);
    ($L_1$,$b$,$R_1$) = split($T_1$,$k_2$);
    $\Tl$ = difference($L_1$,$L_2$) || $\Tr$ = difference($R_1$,$R_2$);
    join2($\Tl$,$\Tr$) @\vspace{.1in}@
\end{lstlisting}
\end{minipage}

\hrulefill

\caption{ Implementing \union{}, \intersect{}, \difference{},
  \insertnew, \delete, \split, and \joinTwo{} with just \join.
  \expose{} returns the left tree, key, and right tree of a node.  The
  $||$ notation indicates the recursive calls can run in parallel.
  These are slight variants of the algorithms described by
  Adams~\cite{adams1992implementing}, although he did not consider
  parallelism.}
\label{fig:unioncode}
\end{figure*}

Adams later described very elegant algorithms for union, intersection,
and difference, as well as other functions using weight-balanced trees,
based on a single
function, \join{}~\cite{adams1992implementing,adams1993functional}
(see Figure~\ref{fig:unioncode}).  The algorithms are naturally
parallel.  The \join$(L,k,R)$ function takes a key $k$ and two ordered
sets $L$ and $R$ such that $L < k < R$ and returns the union of the
keys~\cite{Tarjan83,ST85}.  \join{} can be used to implement
\joinTwo$(L,R)$, which does not take the key in the middle, and
\split$(T,k)$, which splits a tree at a key $k$ returning the two
pieces and a flag indicating if $k$ is in $T$ (See Section \ref{sec:operations}).  With these three
functions, union, intersection, and difference (as well as insertion,
deletion and other functions) are almost trivial.  Because of this at
least three libraries use Adams' algorithms for their implementation
of ordered sets and tables (Haskell~\cite{marlow2010haskell} and
MIT/GNU Scheme, and SML).

\join{} can be implemented on a variety of different balanced tree
schemes.  Sleator and Tarjan describe an algorithm for \join{} based
on splay trees which runs in amortized logarithmic time~\cite{ST85}.
Tarjan describes a version for red-black tree that runs in worst case
logarithmic time~\cite{Tarjan83}.  Adams describes version based on
weight-balanced trees~\cite{adams1992implementing}.\footnote{Adams'
  version had some bugs in maintaining the balance, but these were
  later fixed~\cite{hirai2011balancing,straka2012adams}.}  Adams' algorithms
were proposed in an international competition for the Standard ML community, which
is about implementations on ``set of integers''. Prizes were awarded in two categories:
fastest algorithm, and most elegant yet still efficient program. Adams won the elegance award, while
his algorithm is as fast as the fastest program for very large sets, and was faster for smaller sets.
Adams' algorithms actually show that in principle all balance criteria for
search trees can be captured by a single function \join{}, although he only considered weight-balanced trees.

Surprisingly, however, there have been almost no results on bounding
the work (time) of Adams' algorithms, in general nor on specific
trees.  Adams informally argues that his algorithms take $O(n + m)$
work for weight-balanced tree, but that is a very loose bound.  Blelloch and Miller later show
that similar algorithms for treaps~\cite{BR98}, are optimal for work
(i.e. $\Theta\left(\boundcontent\right)$), and also
parallel.  Their algorithms, however, are specific for treaps.  The
problem with bounding the work of Adams' algorithms, is that just
bounding the time of \split{}, \join{} and \joinTwo{} with logarithmic
costs is not sufficient.\footnote{Bounding the cost of \join{},
  \split{}, and \joinTwo{} by the logarithm of the \emph{smaller tree} is
  probably sufficient, but implementing a data structure with such
  bounds is very much more complicated.}  One needs additional
properties of the trees.

The contribution of this paper is to give first work-optimal bounds for
Adams' algorithms. We do this not only for the weight-balanced trees.
we bound the work and depth of Adams' algorithms (union,
intersection and difference) for four different balancing shemes: AVL
trees, red-black trees, weight-balanced trees and treaps.  We analyze
exactly the algorithms in Figure~\ref{fig:unioncode}, and the bounds
hold when either input tree is larger.  We show that with appropriate
(and simple) implementations of \join{} for each of the four balancing
schemes, we achieve asymptotically optimal bounds on work.
Furthermore the algorithms have $O(\log n \log m)$ depth, and hence
are highly parallel.
To prove the bounds on work we show that our
implementations of \join{} satisfy certain conditions based on a rank
we define for each tree type.  In particular the cost of \join{} must
be proportional to the difference in ranks of two trees, and the rank
of the result of a join must be at most one more than the maximum rank
of the two arguments.

In addition to the theoretical analysis of the algorithms, we
implemented parallel versions of all of the algorithms on all four
tree types and describe a set of experiments.  Our implementation is
generic in the sense that we use common code for the algorithms in
Figure~\ref{fig:unioncode}, and only wrote specialized code for each
tree type for the \join{} function.  Our implementations of \join{}
are as described in this paper.  We compare performance across a
variety of parameters.  We compare across the tree types, and
interestingly all four balance criteria have very similar
performance.  We measure the speedup on up to 64 cores and achieve close
to a 46-fold speedup.  We compare to other implementations, including
the set implementation in the C++ Standard Template library (STL) for
sequential performance, and parallel weight-balanced B-trees
(WBB-trees) \cite{WBTree} and the multi-core standard template library
(MCSTL) \cite{FS07} for parallel performance.  We also compare for different data
distributions.   The conclusion from the experiments is that although
not always as fast as (WBB-trees) \cite{WBTree} on uniform distributions, the generic code is
quite competitive, and on keys
with a skewed overlap (two Gaussians with different means), our implementation
is much better than all the other baselines.

\newcommand{\tr}[2]{T_{#1}^{#2}}

\paragraph{Related Work}
Parallel set operations on two ordered sets have been well-studied,
but each previous algorithm only works on one type of balance tree.
Paul, Vishkin, and Wagener studied bulk insertion and deletion on 2-3
trees in the PRAM model~\cite{PVW83}.  Park and Park showed similar
results for red-black trees~\cite{PP01}.  These algorithms are not
based on \join{}.  Katajainen~\cite{katajainen1994efficient} claimed an
algorithm with \bound{} work and $O(\log n)$ depth using 2-3 tree, but
was shown to contain some bugs so the bounds do not hold \cite{BR98}.
Akhremtsev and Sanders~\cite{akhremtsevsanders} (unpublished) recently
fixed that and proposed an algorithm based on $(a,b)$-trees with
optimal work and $O(\log n)$ depth. Their algorithm only works for
$(a,b)$-trees, and they only gave algorithms on \union{}. Besides, our
algorithm is much simpler and easy to be implemented than theirs.
Blelloch and Miller showed a similar algorithm as Adams' (as well as
ours) on treaps with optimal work and $O(\log n)$ depth on a EREW PRAM
with scan operation using pipelines (implying $O(\log n\log m)$ depth
on a plain EREW PRAM, and $O(\log^* m \log n)$ depth on a plain CRCW
PRAM).  The pipelining is quite complicated.  Our focus in this paper
is in showing that very simple algorithm are work efficient and have
polylogarithmic depth, and less with optimizing the depth.

Many researchers have considered concurrent implementations of
balanced search trees (e.g.,~\cite{KL80,Lersen00,BCCO10,NM14}).  None
of these is work efficient for \union{} since it is necessary to
insert one tree into the other taking work at least $O(m \log n)$.
Furthermore the overhead of concurrency is likely to be very high.

\section{Preliminaries}
\label{sec:prelim}
A \emph{binary tree} is either a \leaf{}, or a node consisting of a
\emph{left} binary tree $\Tl$, a value (or key) $v$, and a
\emph{right} binary tree $\Tr$, and denoted \node$(\Tl,v,\Tr)$.  The
\emph{size} of a binary tree, or $|T|$, is $0$ for a \leaf{} and
$|\Tl| + |\Tr|+1$ for a \node$(\Tl,v,\Tr)$.  The \emph{weight} of a
binary tree, or $w(T)$, is one more than its size (i.e., the number of
leaves in the tree).  The \emph{height} of a binary tree, or $h(T)$,
is $0$ for a \leaf, and $\max(h(\Tl), h(\Tr)) + 1$ for a
\node$(\Tl,v,\Tr)$.  \emph{Parent}, \emph{child}, \emph{ancestor} and
\emph{descendant} are defined as usual (ancestor and descendant are
inclusive of the node itself).
%A node has \emph{depth} $d$ if taking its parent $d$ times returns the root.
The \emph{left spine} of a binary tree is the path of nodes from the root to a leaf always
following the left tree, and the \emph{right spine} the path to a leaf
following the right tree.  The \emph{in-order values} of a binary tree
is the sequence of values returned by an in-order traversal of the
tree.

A \emph{balancing scheme} for binary trees is an invariant (or set of
invariants) that is true for every node of a tree, and is for the
purpose of keeping the tree nearly balanced.  In this paper we
consider four balancing schemes that ensure the height of every tree
of size $n$ is bounded by $O(\log n)$.  For each balancing scheme we
define the \emph{rank} of a tree, or $r(T)$.

\textbf{AVL trees}~\cite{avl} have the invariant that
for every $\node(\Tl,v,\Tr)$, the height of $\Tl$ and $\Tr$ differ by
at most one.  This property implies that any AVL tree of size $n$ has
height at most $\log_{\phi}(n+1)$, where $\phi =
\frac{1+\sqrt{5}}{2}$ is the golden ratio.    For AVL trees
$r(T) = h(T)-1$.

\textbf{Red-black (RB) trees}~\cite{redblack} associate a color with
every node and maintain two invariants: (the red rule) no red node has
a red child, and (the black rule) the number of black nodes on every
path from the root down to a leaf is equal.  Unlike some other
presentations, we do not require that the root of a tree is black.
Our proof of the work bounds requires allowing a red root.  We define
the \emph{black height} of a node $T$, denoted $\bh(T)$ to be the
number of black nodes on a downward path from the node to a leaf
(inclusive of the node).  Any RB tree of size $n$ has height at most
$2\log_2 (n+1)$.  In RB trees $r(T)=2(\bh(T)-1)$ if $T$ is black and
$r(T)=2\bh(T)-1$ if $T$ is red.

\textbf{Weight-balanced (WB) trees} with parameter $\alpha$ (also called
BB$[\alpha]$ trees)~\cite{weightbalanced} maintain for every
$T=\node(\Tl,v,\Tr)$ the invariant
$\alpha \le \frac{w(\Tl)}{w(T)}\le 1-\alpha$.  We say two
weight-balanced trees $T_1$ and $T_2$ have \emph{like} weights if
$\node(T_1,v,T_2)$ is weight balanced.  Any $\alpha$ weight-balanced
tree of size $n$ has height at most $\log_{\frac{1}{1-\alpha}}n$.  For
$\frac{2}{11} < \alpha \le 1-\frac{1}{\sqrt{2}}$ insertion and
deletion can be implemented on weight balanced trees using just single
and double rotations~\cite{weightbalanced,blum1980average}.  We
require the same condition for our implementation of \join{}, and in
particular use $\alpha=0.29$ in experiments.
For WB trees $r(T) = \lceil\log_2(w(T))\rceil-1$.

\textbf{Treaps}~\cite{SA96} associate a uniformly random
priority with every node and maintain the invariant that the priority
at each node is no greater than the priority of its two children.  Any
treap of size $n$ has height $O(\log n)$ with high probability
(w.h.p)\footnote{Here w.h.p. means that height $O(c\log n)$ with probability at least $1 -
1/n^c$ ($c$ is a constant)}.   For treaps $r(T) = \lceil\log_2(w(T))\rceil-1$.

For all the four balancing schemes $r(T)=\Theta(\log(|T|+1))$.
The notation we use for binary trees is summarized in
Table~\ref{fig:notation}.  %For our generic functions we only use
%\join{}, \expose{} and $r(T)$ to access the balanced trees.

\begin{comment}
We also denote $\beta = \frac{1-\alpha}{\alpha}$, which means that
either subtree could have a size of more than $\beta$ times of the
other subtree.
\end{comment}

A \emph{Binary Search Tree} (BST) is a binary tree in which each value
is a key taken from a total order, and for which the in-order values
are sorted.  A \emph{balanced BST} is a BST maintained with a
balancing scheme, and is an efficient way to represent ordered sets.

%We note that the balancing schemes
%defined above although typically applied to BSTs do not require that
%the binary tree be a BST.

%When describing operations on a pair of balanced BST (e.g. union or
%intersection) we use $n$ and $m$ to denote the size of the two trees
%where we assume $m \ge n$.

Our algorithms are based on nested parallelism with nested fork-join
constructs and no other synchronization or communication among
parallel tasks.\footnote {This does not preclude using our algorithms
  in a concurrent setting.}  All algorithms are deterministic.  We use
work ($W$) and span ($S$) to analyze asymptotic costs, where the work
is the total number of operations and span is the critical path.  We
use the simple composition rules $W(e_1~||~e_2) = W(e_1) + W(e_2) + 1$
and $S(e_1~||~e_2) = \max(S(e_1),S(e_2)) + 1$.  For sequential
computation both work and span compose with addition.  Any computation
with $W$ work and $S$ span will run in time $T < \frac{W}{P} + S$
assuming a PRAM (random access shared memory) with $P$ processors and
a greedy scheduler~\cite{Brent74,BL98}.

%For different balancing schemes, we define the \emph{rank} of a tree node, denoted as $r(\cdot)$. For AVL nodes, the rank is defined as the height of the node. For red-black trees, the rank is defined as the black height of the node. For treap, it is simply defined as the priority of the node, and for weight-balanced trees, it is defined as the log of size of the subtree. The rank of a tree $T$ is defined as the rank of the tree root. For AVL trees, red-black trees and weight-balanced trees, $r(T)=O(\log |T|)$.

\begin{comment}
We use a function \textal{Connect3}$(\Tl,k,\Tr)$, where $\Tl$ and $\Tr$ are two BBSTs and $k$ is one node of a BBST, to return a new balanced BST of root $k$ whose left child is $\Tl$ and right child is $\Tr$. This requires that all keys in $\Tl$ is smaller than $k$ and all keys in $\Tr$ is larger than $k$. Tab. \ref{notations} shows some notations used in our paper.
% We also use $h(T)$ to denote the height of BST $T$.
\end{comment}

\begin{table}
  \centering
  \begin{tabular}{c|c}
    \hline
    % after \\: \hline or \cline{col1-col2} \cline{col3-col4} ...
    \textbf{Notation} & \textbf{Description} \\
    \hline
    $|T|$ & The size of tree $T$\\
    $h(T)$ & The height of tree $T$ \\
    $\bh(T)$ & The black height of an RB tree $T$ \\
    $r(T)$ & The rank of tree $T$ \\
    $w(T)$ & The weight of tree $T$ (i.e, $|T| + 1$)\\
    $p(T)$ & The parent of node $T$\\
    $k(T)$ & The value (or key) of node $T$ \\
    %\hline
    $L(T)$ & The left child of node $T$ \\
    %\hline
    $R(T)$ & The right child of node $T$ \\
    \texttt{expose}$(T)$ & $(L(T),k(T),R(T))$ \\
    %\hline
    %$k.{\rm{data}}$ & The data stored in node $k$ \\
    %$k_1<k_2$ & $k_1.{\rm{data}}<k_2.{\rm{data}}$ \\
    %$T_1<T_2$ & $\forall k_1 \in T_1, k_2\in T_2$, $k_1<k_2$ \\
    %$T_1<k$ & $\forall k_1 \in T_1$, $k_1<k$ \\
    \hline
  \end{tabular}
  \caption{Summary of notation.}\label{fig:notation}
\end{table}

\section{The JOIN Function}
\label{sec:paradigm}
Here we describe algorithms for \join{} for the four balancing schemes
we defined in Section~\ref{sec:prelim}.  %We use a generalized version
%of \join{} that works for any binary tree, whether a BST or not.
The function \join$(\Tl,k,\Tr)$ takes two binary trees
$\Tl$ and $\Tr$, and a value $k$, and returns a new binary tree for
which the in-order values are a concatenation of the in-order values
of $\Tl$, then $k$, and then the in-order values of $\Tr$.  %If the
%binary trees are BSTs, then this is the same as the definition in the
%introduction.  %However, binary trees can also be used to represent
%sequences, and \join{} might be useful for sequences.

%is useful in this context as describe at the end of Section~\ref{sec:??}.

As mentioned in the introduction and shown in Section \ref{sec:operations}, \join{} fully captures what is
required to rebalance a tree and can be used as the only function that
knows about and maintains the balance invariants.  For AVL, RB and WB
trees we show that \join{} takes work that is proportional to the
difference in rank of the two trees.  For treaps the work depends on
the priority of $k$.  All versions of \join{} are sequential so the
span is equal to the work.  Due to space limitations, we
describe the algorithms, state the theorems for all balancing schemes,
but only show a proof outline for AVL trees.
%We show in the next section that all
%balancing schemes we consider lead to optimal work algorithms for many
%other functions on maps and sets.

\begin{figure}[!h!t]
\small
\begin{lstlisting}[frame=lines]
joinRight$(\Tl,k,\Tr)$ =
  $(l,k',c)$ = expose$(\Tl)$;
  if $h(c) \leq h(\Tr) + 1$ then @\label{line:avlbase}@
    $T' = \node(c,k,\Tr)$; @\label{line:avlnew}@
    if $h(T') \leq h(l) + 1$ then $\node(l ,k', T')$
    else rotateLeft($\node$($l, k',$ rotateRight$(T')$)) @\label{line:avldouble}@
  else
    $T'$ = joinRight$(c,k,\Tr)$;
    $T''$ = $\node(l, k',T')$;
    if $h(T') \leq h(l) + 1$ then $T''$
    else rotateLeft$(T'')$@\vspace{.1in}@ @\label{line:avlup}@
join$(\Tl,k,\Tr)$ =
  if $h(\Tl) > h(\Tr) + 1$ then joinRight$(\Tl,k,\Tr)$ @\label{line:avljoinright}@
  else if $h(\Tr) > h(\Tl) + 1$ then joinLeft$(\Tl,k,\Tr)$
  else $\node(\Tl,k,\Tr)$
\end{lstlisting}
\caption{AVL \join{} algorithm.}
\label{fig:avljoin}
\end{figure}

\begin{figure}[!h!t]
\small
\begin{lstlisting}[frame=lines]
joinRightRB$(\Tl,k,\Tr)$ =
  if ($r(\Tl) = \lfloor r(\Tr)/2 \rfloor\times 2$) then
    $\node(\Tl,\left<k,\texttt{red}\right>,\Tr)$;
  else
    $(L',\left<k',c'\right>,R')$=expose($\Tl$);
    $T'$ = $\node(L', \left<k',c'\right>$,joinRightRB$(R',k,\Tr))$;
    if ($c'$=black) and ($c(R(T'))=c(R(R(T')))$=red) then
      $c(R(R(T')))$=black;
      $T''$=rotateLeft$(T')$
    else $T''$@\vspace{.1in}@
joinRB$(\Tl,k,\Tr)$ =
  if $\lfloor r(\Tl)/2 \rfloor > \lfloor r(\Tr)/2 \rfloor$ then
    $T'=$joinRightRB$(\Tl,k,\Tr)$;
    if ($c(T')$=red) and ($c(R(T'))$=red) then
      $\node(L(T'),\left<k(T'),\texttt{black}\right>,R(T'))$
    else $T'$
  else if $\lfloor r(\Tr)/2 \rfloor > \lfloor r(\Tl)/2 \rfloor$ then
    $T'=$joinLeftRB$(\Tl,k,\Tr)$;
    if ($c(T')$=red) and ($c(L(T'))$=red) then
      $\node(L(T'),\left<k(T'),\texttt{black}\right>,R(T'))$
    else $T'$
  else if ($c(\Tl)$=black) and ($c(\Tr)$=black) then
    $\node(\Tl,\left<k,\texttt{red}\right>,\Tr)$
  else $\node(\Tl,\left<k,\texttt{black}\right>,\Tr)$
\end{lstlisting}
\caption{RB \join{} algorithm.}
\label{fig:rbjoin}
\end{figure}

\begin{figure}[!h!t]
\small
\begin{lstlisting}[frame=lines]
joinRightWB$(\Tl,k,\Tr)$ =
  ($l,k',c$)=expose($\Tl$);
  if (balance($|\Tl|,|\Tr|$) then $\node(\Tl,k,\Tr))$;
  else
    $T'$ = joinRightWB$(c,k,\Tr)$;
    $(l_1,k_1,r_1)$ = expose$(T')$;
    if like$(|l|,|T'|)$ then $\node$($l, k', T'$)
    else if (like$(|l|,|l_1|)$) and (like$(|l|+|l_1|,r_1)$) then
      rotateLeft($\node$($l, k', T'$))
    else rotateLeft($\node$($l,k'$,rotateRight$(T')$))@\vspace{.1in}@
joinWB$(\Tl,k,\Tr)$ =
  if heavy($\Tl, \Tr$) then joinRightWB$(\Tl,k,\Tr)$
  else if heavy($\Tr,\Tl$) then joinLeftWB$(\Tl,k,\Tr)$
  else $\node(\Tl,k,\Tr)$
\end{lstlisting}
\caption{WB \join{} algorithm.}
\label{fig:wbjoin}
\end{figure}

\begin{figure}[!h!t]
\small
\begin{lstlisting}[frame=lines]
joinTreap$(\Tl,k,\Tr)$ =
  if prior($k, k_1$) and prior($k,k_2$) then $\node(\Tl,k,\Tr)$
  else ($l_1,k_1,r_1$)=expose($\Tl$);
    ($l_2,k_2,r_2$)=expose($\Tr$);
    if prior($k_1,k_2$) then
      $\node$($l_1,k_1$,joinTreap$(r_1,k,\Tr)$)
    else $\node$(joinTreap($\Tl,k,l_2$),$k_2,r_2$)
\end{lstlisting}
\caption{Treap \join{} algorithm.}
\label{fig:treapjoin}
\end{figure}

\para{AVL trees.}
Pseudocode for AVL \join{} is given in Figure~\ref{fig:avljoin} and
illustrated in Figure~\ref{avlrebalance}.  Every node stores its own
height so that $h(\cdot)$ takes constant time.  If the two trees $\Tl$
and $\Tr$ differ by height at most one, \join{} can simply create a
new \node$(\Tl,k,\Tr)$.  However if they differ by more than one then
rebalancing is required.  Suppose that $h(\Tl) > h(\Tr) + 1$ (the
other case is symmetric).  The idea is to follow the right spine of
$\Tl$ until a node $c$ for which $h(c) \leq h(\Tr) + 1$ is found
(line~\ref{line:avlbase}).  At this point a new $\node(c,k,\Tr)$ is
created to replace $c$ (line~\ref{line:avlnew}).  Since either $h(c) =
h(\Tr)$ or $h(c) = h(\Tr) + 1$, the new node satisfies the AVL
invariant, and its height is one greater than $c$.  The increase in
height can increase the height of its ancestors, possibly invalidating
the AVL invariant of those nodes.  This can be fixed either with a
double rotation if invalid at the parent (line~\ref{line:avldouble})
or a single left rotation if invalid higher in the tree
(line~\ref{line:avlup}), in both cases restoring the height for any
further ancestor nodes.  The algorithm will therefore require at most
two rotations.

\begin{figure*}[t!h!]
\centering
  \includegraphics[width=0.9\columnwidth]{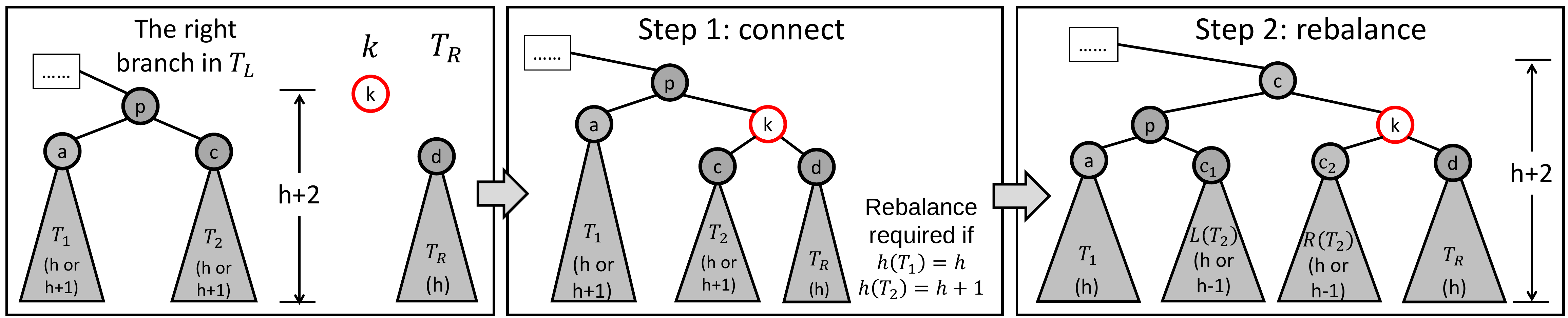}\\
  \caption{An example for \join{} on AVL trees ($h(\Tl) > h(\Tr)+1$).
    We first follow the right spine of $\Tl$ until a subtree of height at most $h(T_r) +
    1$ is found (i.e., $T_2$ rooted at $c$).  Then a new \node$(c,k,\Tr)$
    is created, replacing $c$ (Step 1).    If $h(T_1) = h$ and $h(T_2)
    = h + 1$, the node $p$ will no longer satisfy the AVL invariant.
    A double rotation (Step 2) restores both balance and
    its original height.}
\label{avlrebalance}
\end{figure*}

\begin{lemma}
\label{lem:AVL}
For two AVL trees $\Tl$ and $\Tr$, the AVL \join{} algorithm works
correctly, runs with $O(|h(\Tl)-h(\Tr)|)$ work, and returns a
tree satisfying the AVL invariant with height at most
$1 +\max(h(\Tl),h(\Tr))$.
\end{lemma}

\proofoutline{} Since the algorithm only visits nodes on the path from
the root to $c$, and only requires at most two rotations, it does work
proportional to the path length.   The path length
is no more than the difference in height of the two trees since
the height of each consecutive node along the right spine of $\Tl$
differs by at least one.  Along with the case when
$h(\Tr) > h(\Tl) + 1$, which is symmetric, this gives the stated work
bounds.  The resulting tree satisfies the AVL
invariants since rotations are used to restore the invariant (details
left out).    The height of any node can increase by at most
one, so the height of the whole tree can increase by at most one.
\qed

\fullin{
\proof
For symmetry, here we only prove the case when $h(\Tl)>h(\Tr)$. In our algorithm, we will first go right down to a proper level then apply rebalance along the path upwards. We first prove that after applying rebalance at some node $v$, $v$ will be AVL-balanced, and the height of the subtree rooted at $v$ will increase at most $1$. Recall that we denote the node that we stop going down in $\Tl$ as $c$, its parent as $p$, and $p$'s parent as $g$. Suppose $h(\Tr)=h$. We prove this theorem by induction.

As shown in Fig. \ref{avlrebalance}, after appending $\Tr$ to $\Tl$ by the intermediate key $k$, $k$ is guaranteed to be balanced, and the subtree rooted at $k$ has a height of $h+1$. We discuss the three possible scenarios based on the height of $p$' left tree (i.e., $T_1$ rooted at $a$ as shown in Fig. \ref{avlrebalance}), i.e., $h-1$, $h$ and $h+1$. The only case that an unbalance will occur at $p$ is when the height of the left child of $p$ is $h-1$. We apply a left rotation here and the subtree originally rooted at $p$ is now rooted at $k$, and is now balanced. The whole subtree height is now $h+2$ instead of original $h+1$. If $p$'s left child has a height of $h$, the height of the original subtree is $h+1$. In this case, $p$ is now balanced, and the height is currently $h+2$. If the height of the left child of $p$ is $h+1$, $p$ is also naturally balanced, and the height of $p$ does not increase. Thus in all three scenarios the proposition holds at the very first linking of $\Tl$ and $\Tr$. If we have reached the root, the rebalance stops and the current tree is a valid AVL. Otherwise we call rebalance for $p$'s parent.

When we reach some interior nodes $v$, $v$ is originally balanced, and the height of the right child (denoted as $u$) of $v$ increases by at most $1$. If $v$ is still balanced, the whole height of $v$ is increased by at most $1$. Suppose originally the height of the right child of $v$ is $h$. The only case that $v$ is now unbalanced is when $v$'s left child's height is $h-1$, and the current height of $v$'s right child (i.e., $u$) is $h+1$. Similar as stated above, based on the height of $u$'s left child ($h_L$) and right child ($h_R$), we apply the following rotations:
\begin{center}
\begin{tabular}{c|c|c|c|c}
  \hline
  \multirow{2}{*}{$h_L$} & \multirow{2}{*}{$h_R$} & \multirow{2}{*}{Operation} & Original & Final\\
  &&&Height & Height\\
  \hline
  $h$ & $h$ &  Left rotation & $h+1$ & $h+2$\\
  \hline
  \multirow{2}{*}{$h$} & \multirow{2}{*}{$h-1$} & Right rotation  & \multirow{2}{*}{$h+1$} & \multirow{2}{*}{$h+1$} \\
  &&then left rotation&\\
  \hline
  $h-1$ & $h$ & Left rotation & $h+1$ & $h+1$ \\
  \hline
\end{tabular}
\end{center}
In each case the subtree originally at $v$ will be balanced and the height will increase by at most $1$. If we have reached the root, the rebalance stops and the current tree is a valid AVL. Otherwise we call rebalance for $v$'s parent.

As the process will guarantee that current node to be rebalanced will go up one level, we finally will reach the root and end rebalance process, hence get a valid AVL.
\qed
}

\para{Red-black Trees.}
Tarjan describes how to implement the \join{} function for red-black
trees~ \cite{Tarjan83}.  Here we describe a variant that does not
assume the roots are black (this is to bound the increase in rank
by \union{}).
The pseudocode is given in Figure \ref{fig:rbjoin}.
We store at every node its black height
$\bh(\cdot)$.  The first case is when $\bh(\Tr)=\bh(\Tl)$.  Then if
both $k(\Tr)$ and $k(\Tl)$ are black, we create red
$\node(\Tl,k,\Tr)$, otherwise we create black $\node(\Tl,k,\Tr)$.  The
second case is when $\bh(\Tr)<\bh(\Tl)=\bh$ (the third case is symmetric).
Similarly to AVL trees, \join{} follows the right spine of $\Tl$ until
it finds a black node $c$ for which $\bh(c) = \bh(\Tr)$.  It then
creates a new red \node$(c,k,\Tr)$ to replace $c$. Since both $c$
and $\Tr$ have the same height, the only invariant that can be
violated is the red rule on the root of $\Tr$, the new node, and its
parent, which can all be red.  In the worst case we may have three red
nodes in a row.  This is fixed by a single left rotation: if a black
node $v$ has $R(v)$ and $R(R(v))$ both red, we turn $R(R(v))$ black
and perform a single left rotation on $v$. %turning the new node black, and then performing a single left rotation on $v$.
The update is
illustrated in Figure~\ref{rbrebalance}.  The rotation, however can
again violate the red rule between the root of the rotated tree and
its parent, requiring another rotation.
% but expect the bottommost level, a triple-red issue does not happen.
The double-red issue might
proceed up to the root of $\Tl$. If the original root of $\Tl$ is red,
the algorithm may end up with a red root with a red child, in which
case the root will be turned black, turning $\Tl$ rank from $2\bh-1$
to $2\bh$.  If the original root of $\Tl$ is black, the algorithm
may end up with a red root with two black children, turning the rank
of $\Tl$ from $2\bh-2$ to $2\bh-1$.  In both cases the rank of the
result tree is at most $1+r(\Tl)$.

\hide{Similarly to AVL trees, the algorithm
follows the right spine of $\Tl$ until it finds a black node $c$ for
which $\bh(c) = \bh(\Tr)$.  It then creates a new red \node$(c,k,\Tr)$
to replace $c$.  Since both $c$ and $\Tr$ are black and have the same
height, the only invariant that can be violated is the red rule
between the new node and its parent, which can be red.  This is fixed
by turning the new node black, and then performing a single left
rotation on its grandparent to restore the black rule.  The update is
illustrated in Figure~\ref{rbrebalance}.  The rotation, however can
again violate the red rule between the root of the rotated tree and
its parent, requiring another rotation.  This
might proceed up to the root of $\Tl$, which is turned black if red,
therefore possibly incrementing the overall black height by one.
}

\begin{lemma}
\label{lem:RB}
For two RB trees $\Tl$ and $\Tr$, the RB \join{} algorithm works
correctly, runs with $O(|r(\Tl)-r(\Tr)|)$ work, and
returns a tree satisfying the red-black invariants and with rank
at most $1+\max(r(\Tl),r(\Tr))$.
\end{lemma}

The proof is similar as Lemma \ref{lem:AVL}.
\hide{
\proofoutline{} Since the algorithm only visits nodes on the path from
the root to $c$, and only requires at most a single rotation per node
on the path, the overall work for the algorithm is proportional to the
depth of $c$ in $\Tr$.  This in turn is no more than twice the
difference in black height of the two trees since the black height
decrements at least every two nodes along the path.  Along with the
case when $\bh(\Tr) > \bh(\Tl)$, which is true by symmetry, this gives the stated
work bounds.  The algorithm returns a tree that satisfies the
RB invariants since rotations are used to restore the invariant on
every visited node. \qed}

\fullin{
\proof
For symmetry, here we only prove the case when $h(\Tl)>h(\Tr)$. We prove the proposition by induction.

As shown in Fig. \ref{rbrebalance}, after appending $\Tr$ to $\Tl$, if $p$ is black, the rebalance has been done, the height of each node stays unchanged. Thus the RB tree is still valid. Otherwise, $p$ is red, $p$'s parent $g$ must be black. By applying a left rotation on $p$ and $g$, we get a balanced RB tree rooted at $p$, except the root $p$ is red. If $p$ is the root of the whole tree, we change $p$'s color to black, and the height of the whole tree increases by 1. The RB tree is still valid. Otherwise, if the current parent of $p$ (originally $g$'s parent) is black, the rebalance is done here. Otherwise a similar rebalance is required on $p$ and its current parent. Thus finally we will either find the current node valid (current red node has a black parent), or reach the root, and change the color of root to be black. Thus when we stop, we will always get a valid RB tree. And the black height of the whole tree will increase by at most 1.
\qed
}

\begin{figure}[t!h!]
\centering
  \includegraphics[width=0.6\columnwidth]{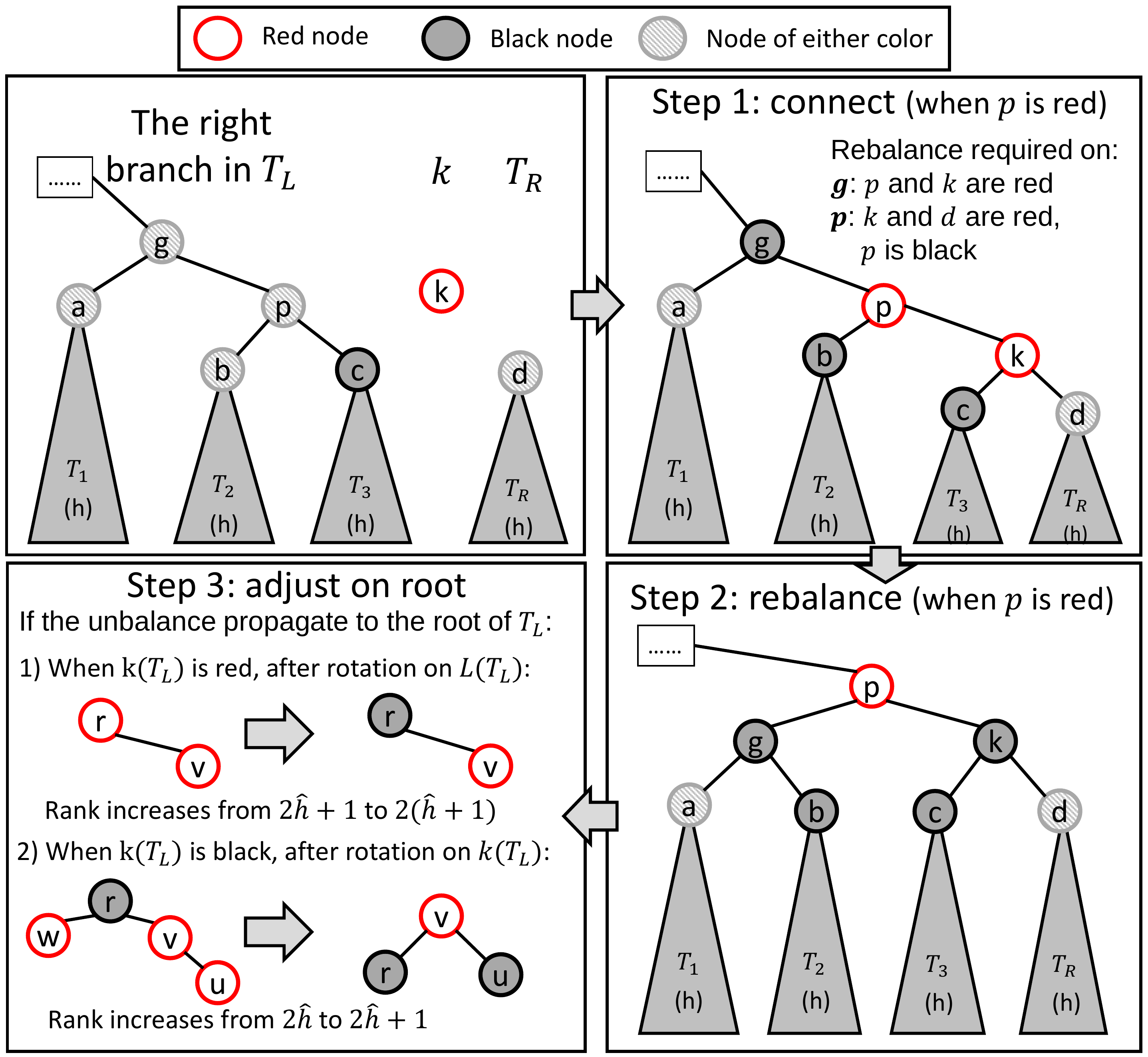}\\
  \caption{An example of \join{} on red-black trees ($\bh=\bh(\Tl) >
    \bh(\Tr)$).  We follow the right spine of $\Tl$ until we find a
    black node with the same black height as $\Tr$ (i.e., $c$).  Then
    a new red \node$(c, k, \Tr)$ is created, replacing $c$ (Step 1).
    The only invariant that can be violated is when either $c$'s
    previous parent $p$ or $\Tr$'s root $d$ is red.  If so, a left
    rotation is performed at some black node. Step 2 shows the
    rebalance when $p$ is red. The black height of the rotated subtree
    (now rooted at $p$) is the same as before ($h+1$), but the parent
    of $p$ might be red, requiring another rotation.  If the red-rule
    violation propagates to the root, the root is either colored red,
    or rotated left (Step 3).}
\label{rbrebalance}
\end{figure}

\para{Weight Balanced Trees.}
We store the weight of each subtree at every node.  The algorithm for
joining two weight-balanced trees is similar to that of AVL trees and
RB trees. The pseudocode is shown in Figure \ref{fig:wbjoin}.
The \texttt{like} function in the code returns true if the two input tree
sizes are balanced, and false otherwise.
If $\Tl$ and $\Tr$ have like weights the algorithm
returns a new \node$(\Tl,k,\Tr)$.  Suppose $|\Tr|\le|\Tl|$, the
algorithm follows the right branch of $\Tl$ until it reaches a node
$c$ with like weight to $\Tr$.  It then creates a new
$\node(c,k,\Tr)$ replacing $c$.  The new node will have weight
greater than $c$ and therefore could imbalance the weight of $c$'s
ancestors.  This can be fixed with a single or double
rotation (as shown in Figure~\ref{wbtree1}) at each node assuming
$\alpha$ is within the bounds given in Section~\ref{sec:prelim}.
\begin{figure}[!t!h]
  % Requires \usepackage{graphicx}
  \centering
  \includegraphics[width=0.6\columnwidth]{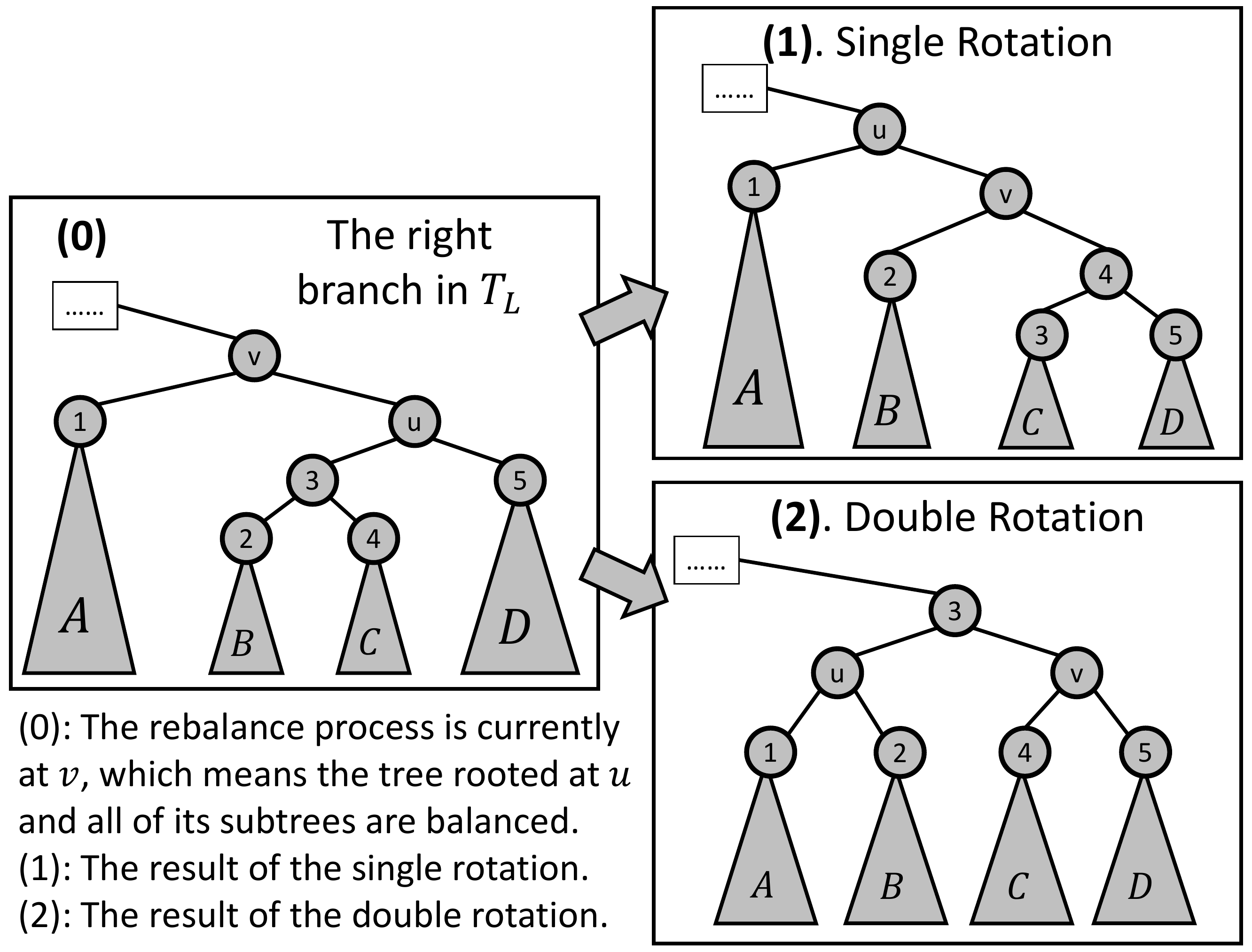}\\
  \caption{An illustration of single and double rotations possibly
    needed to rebalance weight-balanced trees.   In the figure the
    subtree rooted at $u$ has become heavier due to joining in $\Tl$ and its
    parent $v$ now violates the balance invariant.}
\label{wbtree1}
\end{figure}

\begin{lemma}
\label{lem:wbvalid}
For two $\alpha$ weight-balanced trees $\Tl$ and $\Tr$ and
$\alpha \le 1-\frac{1}{\sqrt{2}}\approx 0.29$, the weight-balanced
\join{} algorithm works correctly, runs with $O(|\log(w(\Tl)/w(\Tr))|)$ work, and returns a tree satisfying the $\alpha$ weight-balance invariant and with rank
at most $1+\max(r(\Tl),r(\Tr))$.
\end{lemma}
The proof is shown in the Appendix.

%The proof is intuitively similar as the proof stated in \cite{weightbalanced,blum1980average}, which proved that when $\frac{2}{11}\le \alpha \le 1-\frac{1}{\sqrt{2}}$, the rotation will rebalance the tree after one single insertion. In fact, in our \join{} algorithm, the ``inserted'' subtree must be along the left or right spine, which makes the analysis much easier.
%of the leaf weight is increased to at most twice of its original weight. In our version, the weight of one subtree (in fact, the subtree must be along the left or right spine, which makes the analysis much easier) will be increased to at most $\frac{1}{\alpha}$ times of its original weight. Using similar proof as in \cite{weightbalanced,blum1980average}, the rotation still works.

\begin{comment}
Notice that this upper bound is the same as the restriction on $\alpha$ to yield a valid weighted-balanced tree when inserting a single node.% to make the balance algorithm work when inserting a single node.
Then we can induce that when the rebalance process reaches the root, the new weight-balanced tree is valid.
\end{comment}

%Notice that the upper bound is the same as the restriction to guarantee the balance algorithm to work when inserting a single node. The proof of the theorem is shown in the Appendix.

\para{Treaps.}  The treap \join{} algorithm (as in Figure
\ref{fig:treapjoin}) first picks the key with the highest priority
among $k$, $k(\Tl)$ and $k(\Tr)$ as the root.  If $k$ is the root then
the we can return \node$(\Tl,k,\Tr)$.  Otherwise, WLOG, assume
$k(\Tl)$ has a higher priority.  In this case $k(\Tl)$ will be the
root of the result, $L(\Tl)$ will be the left tree, and $R(\Tl)$, $k$
and $\Tr$ will form the right tree.  Thus \join{} recursively calls
itself on $R(\Tl)$, $k$ and $\Tr$ and uses result as $k(\Tl)$'s right
child.  When $k(\Tr)$ has a higher priority the case is symmetric.
The cost of \join{} is therefore the depth of the key $k$ in the
resulting tree (each recursive call pushes it down one level).  In treaps
the shape of the result tree, and hence the depth of $k$, depend only
on the keys and priorities and not the history.  Specifically, if a
key has the $t^{th}$ highest priority among the keys, then its
expected depth in a treap is $O(\log t)$ (also w.h.p.).  If it is the
highest priority, for example, then it remains at the root.

\hide{The treap \join{} algorithm first creates the \node$(\Tl,k,\Tr)$,
and then as long as one of the children of $k$ has a higher priority
than $k$, rotates the tree around $k$ so the highest priority key is at
the root.  The cost of \join{} is therefore the depth of the key $k$
in the result tree (each rotation pushes it down one level).  In
treaps the shape of the result tree, and hence the depth of $k$, depend
only on the keys and priorities and not the history.  Specifically, if
a key has the $t^{th}$ highest priority among the keys, then its
expected depth in a treap is $O(\log t)$.  If it is the highest
priority, for example, then it remains at the root.}

\begin{lemma}
\label{lem:treap}
For two treaps $\Tl$ and $\Tr$, if the priority of $k$ is the $t$-th
highest among all keys in $\Tl \cup \{k\} \cup \Tr$, the treap \join{}
algorithm works correctly, runs with $O(\log t+1)$ work
in expectation and w.h.p., and returns a tree satisfying the treap
invariant with rank at most $1 +\max(r(\Tl),r(\Tr))$.
\end{lemma}

From the above lemmas we can get the following fact for \join.

\begin{theorem}
\label{thm:main}
For AVL, RB and WB trees \join$(\Tl,k,\Tr)$ does $O\left(|r(\Tl)-r(\Tr)|\right)$
work. For treaps \join{} does $O(\log t)$ work in expectation if $k$ has the $t$-th
highest priority among all keys.
For AVL, RB, WB trees and treaps, \join{} returns a tree $T$ for which the rank satisfies
$\max\left(r(\Tl),r(\Tr)\right)\le r(T)\le 1+\max(r(\Tl),r(\Tr))$.
\end{theorem}
\par~\par

\section{Other Functions Using JOIN}
\label{sec:operations}
%The \join{} function, as a subroutine, has been used and studied by many
%researchers and programmers to implement more general set operations.
In this section, we describe algorithms for various functions that use
just \join{}.  The algorithms are generic across balancing schemes.  The pseudocodes for the algorithms in this section is shown in Figure \ref{fig:unioncode}.
%Beyond \join{} the only access to the trees we make use of
%is through \expose{} and $r(\cdot)$, which only read the root. %The
%main set operations, which are \union{}, \intersect{} and \difference{},
%are optimal (or known as \emph{efficient}) in work.
\hide{
Although similar algorithms has been introduced in previous research,
to our best knowledge, there has been no detailed analysis of the efficiency of those
algorithms other than \cite{BR98}, which analyzed the work and span only for treaps.
Besides their algorithm is different from ours in the sense that they always use
the heavier root as the pivot (see details later).
Based on the algorithms and theorems in this section, in Section \ref{sec:proof}, we provide a general
proof, which is very different from and much simpler than \cite{BR98}, that works for all
the four balancing schemes mentioned in this paper.}

\para{Split.}
For a BST $T$ and key $k$, \split{}$(T,k)$ returns a triple $(\Tl,b,\Tr)$, where $\Tl$ ($\Tr$) is a tree containing all keys in $T$ that are less (larger) than $k$, and $b$ is a flag indicating whether $k \in T$.  The algorithm first searches for $k$ in $T$, splitting the tree along the path into three parts: keys to the left of the path, $k$ itself (if it exists), and keys to the right.  Then by applying \join{}, the algorithm merges all the subtrees on the left side (using keys on the path as intermediate nodes) from bottom to top to form $\Tl$, and merges the right parts to form $\Tr$.
Figure~\ref{fig:split} gives an example.
%The pseudocode is shown in Figure \ref{fig:insertandsplit}.

\begin{figure}[t!h!]
\centering
  \includegraphics[width=0.9\columnwidth]{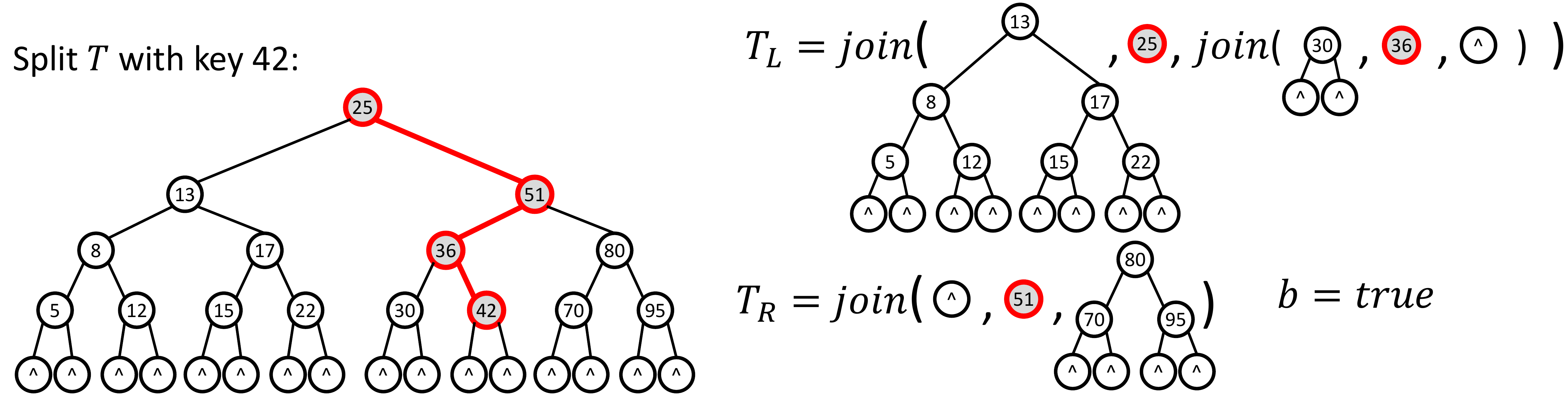}\\
  \caption{An example of \split{} in a BST with key $42$. We first search for $42$ in the tree and split the tree by the searching path, then use \join{} to combine trees on the left and on the right respectively, bottom-top. }
\label{fig:split}
\end{figure}

\begin{theorem}
\label{thm:split}
The work of \split{}$(T,k)$ is $O(\log |T|)$ for all balancing schemes described in Section \ref{sec:paradigm} (w.h.p. for treaps). The two resulting trees $\Tl$ and $\Tr$ will have rank at most $r(T)$.
\end{theorem}
\begin{proof}
We only consider the work of joining all subtrees on the left side.
The other side is symmetric. Suppose we have $l$ subtrees on the left side, denoted from bottom to top as $T_1,T_2,\dots T_l$.  We have that $r(T_1)\le r(T_2)\le\dots\le r(T_l)$. As stated above, we consecutively join $T_1$ and $T_2$ returning $T'_2$, then join $T'_2$ with $T_3$ returning $T'_3$ and so forth, until all trees are merged. The overall work of \split{} is the sum of the cost of $l-1$ \join{} functions.

For AVL trees, red-black trees and weight-balanced trees, recall Theorem \ref{thm:main} that we have $r(T'_{i}) \le r(T_{i})+1$, so $r(T'_i)\le r(T_i)+1\le r(T_{i+1})+1$. According to Theorem \ref{thm:main}, the work of a single operation is $O(|r(T_{i+1})-r(T'_i)|)$. The overall complexity is $\sum_{i=1}^{l} |r(T_{i+1})-r(T'_i)|\le \sum_{i=1}^{l} r(T_{i+1})-r(T'_i)+2=O(r(T))=O(\log |T|)$.
\hide{
\begin{align*}
&\sum_{i=1}^{l} |h(T_{i+1})-h(T'_i)|\\
\le&\sum_{i=1}^{l} h(T_{i+1})-h(T'_i)+4\\
%\le& 4l+\sum_{i=1}^{l} (h(T_{i+1})-h(T_i))\\
\le& 5h(T)= O(h(T))
\end{align*}
}

For treaps, each \join{} uses the key with the highest priority since
the key is always on a upper level.  Hence by Lemma \ref{lem:treap}, the
complexity of each \join{} is $O(1)$ and the work of split is at
most $O(\log |T|)$ w.h.p.

Also notice that when getting the final result $\Tl$ and $\Tr$, the last step is a \join{}
on two trees, the larger one of which is a subtree of the original $T$. Thus the rank of
the two trees to be joined is of rank at most $r(T)-1$, according to Theorem \ref{thm:main}
we have $r(\Tl)$ and $r(\Tr)$ at most $r(T)$.
\end{proof}

%Notice that in the process of splitting a treap, when we call \join{}, the key is always of a higher priority since it is on a upper level than the two subtrees. Thus \join{} could be done in $O(1)$ time.

\para{Join2.}
\joinTwo{}$(\Tl,\Tr)$ returns a binary tree for which the in-order
values is the concatenation of the in-order values of the binary trees
$\Tl$ and $\Tr$ (the same as \join{} but without the middle key).
For BSTs, all keys in $\Tl$ have to be less than keys in $\Tr$.
\joinTwo{} first finds the last element $k$ (by following the right spine) in $\Tl$ and
on the way back to root, joins the subtrees along the path, which is similar to \split{}
$\Tl$ by $k$. We denote the result of dropping $k$ in $T_L$ as $\Tl'$. Then \join($\Tl',k,\Tr$)
is the result of \joinTwo. Unlike \join{}, the work of \joinTwo{} is
proportional to the rank of both trees since both \split{} and \join{} take at most logarithmic work.
%log of size of both trees since both \split{} and \join{} take at most logarithmic work..

\hide{$(\Tl,\Tr)$ returns a binary tree for which the in-order
values is the concatenation of the in-order values of the binary trees
$\Tl$ and $\Tr$ (the same as \join{} but without the middle value).
For BSTs, all keys in $\Tl$ have to be less than keys in $\Tr$.  The
algorithm is recursive.  Assume $\Tl$ has a higher rank than $\Tr$, it
then exposes the larger tree $\Tl$ into $(L_L,k_L,R_L)$
(line~\ref{line:join2expose}), recursively applies
\joinTwo{}$(R_L,\Tr)$ to get $T'$, and then applies
\join($L_L,k_L,T'$) to merge the three parts
(line~\ref{line:join2join}).
  Unlike \join{}, the work of \joinTwo{} is
proportional to the log of size of both trees.
}
\begin{theorem}
\label{thm:join2}
The work of \joinTwo{}$(\Tl,\Tr)$ is $O(r(\Tl) + r(\Tr))$ for all balancing schemes described in Section \ref{sec:paradigm} (bounds are w.h.p for treaps). %For AVL trees, RB trees and WB trees, $\max(r(\Tl),r(\Tr))\le r(T) \le 1+\max(r(\Tl),r(\Tr))$. %, where $c$ is a constant.
\end{theorem}

\para{Union, Intersect and Difference.}
\union{}$(T_1,T_2)$ takes two BSTs and returns a BST that contains the union
of all keys.   The algorithm uses a classic
divide-and-conquer strategy, which is parallel.   At each level of recursion,
$T_1$ is split by $k(T_2)$, breaking $T_1$ into three parts: one with all keys smaller than $k(T_2)$ (denoted as $L_1$), one in the middle either of only one key equal to $k(T_2)$ (when $k(T_2)\in T_1$) or empty (when $k(T_2) \notin T_1$), the third one with all keys larger than $k(T_2)$ (denoted as $R_1$).
%with keys equivalent to $T_1$'s root (just a bit $b$ to show if this part is empty).
%Meanwhile, all keys in the left (right) subtree of $T_1$ are originally smaller (larger) than $k(T_1)$.
Then two recursive calls to \union{} are made in parallel.  One unions
$L(T_2)$ with $L_1$, returning $\Tl$, and the other one unions
$R(T_2)$ with $R_1$, returning $\Tr$.  Finally the algorithm returns
\join$(\Tl,k(T_2),\Tr)$, which is valid since $k(T_2)$ is greater than
all keys in $\Tl$ are less than all keys in $\Tr$.

The functions \intersect$(T_1,T_2)$ and \difference$(T_1,T_2)$ take
the intersection and difference of the keys in their sets,
respectively.  The algorithms are similar to \union{} in that they use one tree to split the other.  However, the
method for joining is different.  For \intersect{}, \joinTwo{} is used
instead of \join{} if the root of the first \emph{is not} found in the
second.  For \difference{}, \joinTwo{} is used anyway because $k(T_2)$ should be excluded in the result tree.  The base cases are
also different in the obvious way.

\begin{theorem}[Main Theorem]
\label{thm:unionwork}
For all the four balance schemes mentioned in Section \ref{sec:paradigm}, the work and span of the algorithm (as shown in Figure~\ref{fig:unioncode}) of \union{}, \intersect{} or \difference{} on two balanced BSTs of sizes $m$ and $n$ ($n\ge m$) is $\displaystyle O\left(\boundcontent\right)$ and $O(\log n \log m)$ respectively (the bound is in expectation for treaps).
\end{theorem}

A generic proof of Theorem \ref{thm:unionwork}
working for all the four balancing schemes will be shown in the next section.
\fullin{
The span of \textal{Union} can be reduced to $O(\log m)$ by
doing a more complicated divide-and-conquer strategy based on (1)
finding splitters that divide the result into blocks of size $(m \log
m)/(n \log(m/n +1))$, (2) for each block in parallel, extracting the
block from each input with \split{} and taking the union of the results,
and (3) joining the results in parallel using a tree.}

The work bound for these algorithms is optimal in the comparison-based model.
In particular considering all possible interleavings, the minimum
number of comparisons required to distinguish them is $\log
{m+n \choose n}=\Theta\left(\boundcontent\right)$ \cite{hwang1972simple}.

%% The algorithm first uses the root of one tree as the pivot to split the other tree, and recursively intersects the left and right part respectively, getting two trees $\Tl$ and $\Tr$. If the pivot is found in the other tree, it should appear in the return tree, so \join{} is called to join $\Tl$, $\Tr$ and the pivot. Otherwise \joinTwo{} should be called to join $\Tl$ and $\Tr$. The function \difference{} is similar. The only difference is that if the root appears in the other tree, it should be excluded.

%% It is easy to see that the work and depth of \intersect{} and \difference{} is the same as \union.
\hide{
\para{Build.}
\build{} takes a unsorted array and returns a valid BST containing all
keys in that array.  The algorithm divides the full array into two
parts of equal length and recursively build two smaller trees in
parallel, and then unions the results.  %Figure~\ref{fig:pseudocode}
%lines~\ref{line:builds}--\ref{line:builde} shows the pseudocode for
%\build. Here \textal{Singleton} denotes a function to initialize a
%tree with a single node.

\begin{theorem}
The work and span of \build{} on an array of size $n$ is $O(n\log n)$ and $O(\log^3 n)$ respectively for the four balancing schemes mentioned in Section \ref{sec:paradigm}.
\end{theorem}
}
\fullin{
{\bf [Yihan: Are we using depth or span.]}
\proof
We use $W(n)$ and $D(n)$ to denote the work and depth to build a BST from an array of size $n$. From Theorem \ref{mergework} we know that the work and depth of \union{} in the last step is $O(n)$ and $O(\log^2 n)$ respectively. Thus we have:
\begin{align*}
W(n)&=2W(\frac{n}{2})+n\\
%W(1)&=O(1)\\
D(n)&=1+D(\frac{n}{2})+O(\log^2 n)\\
%D(1)&=O(1)
\end{align*}
Considering $W(1)=D(1)=O(1)$, solve the recursion formula we can get $W(n)=O(n\log n)$ and $D(n)=O(\log^3 n)$.
\qed}

\fullin{
Similarly, if the \textal{Union} with $O(\log m)$ depth is applied, the depth of \build{} could be reduced to $O(\log^2 n)$.

\begin{corollary}
The parallel algorithm \build{} is optimal for the four balancing schemes in terms of work.
\end{corollary}
}

\para{Other Functions.}
Many other functions can be implemented with \join{}.  Pseudocode for
\insertnew{} and \delete{} was given in Figure \ref{fig:unioncode}.  For a tree
of size $n$ they both take $O(\log n)$ work.

\section{The Proof of the Main Theorem}
\label{sec:proof}

In this section we prove Theorem \ref{thm:unionwork}, for all the
four balance schemes (AVL trees, RB trees, WB trees and treaps) and
all three set algorithms (\union{}, \intersect{}, \difference{}) from
Figure~\ref{fig:unioncode}.  For this purpose we make two observations.
The first is that all the work for the algorithms can be accounted for
within a constant factor by considering just the work done by the
\split{}s and the \join{}s (or \joinTwo{}s), which we refer to as
\emph{split work} and \emph{join work}, respectively. This is
because the work done between each split and join is constant.  The
second observation is that the split work is identical among the three
set algorithms. This is because the control flow of the three algorithms
is the same on the way down the recursion when doing \split{}s---the
algorithms only differ in what they do at the base case and on the way
up the recursion when they join.

Given these two observations, we prove the bounds on work by first
showing that the join work is asymptotically at most as large as the
split work (by showing that this is true at every node of the
recursion for all three algorithms), and then showing that the split
work for \union{} (and hence the others) satisfies our claimed bounds.

\begin{table}
  \centering
  \setlength{\extrarowheight}{.2em}
\begin{tabular}{c|c}
  \hline
  % after \\: \hline or \cline{col1-col2} \cline{col3-col4} ...
  \textbf{Notation} & \textbf{Description} \\
  \hline
  $T_p$ & The pivot tree \\
  $T_d$ & The decomposed tree \\
  $n$ & $\max(|T_p|,|T_d|)$ \\
  $m$ & $\min(|T_p|,|T_d|)$ \\
  $T_p(v), v \in T_p$ & The subtree rooted at $v$ in $T_p$ \\
  $T_d(v), v\in T_p$ & The tree from $T_d$ that $v$ splits\footnotemark\\
%  $|T_d(v)|$($v\in T_p$)& The splitting size of $v$\\
  $s_{i}$ & The number of nodes in layer $i$ \\
  $v_{kj}$ & The $j$-th node on layer $k$ in $T_p$ \\
  $d(v)$ & The number of nodes attached to a layer\vspace{-.5em}\\
  &root $v$ in a treap \\
  \hline
\end{tabular}
  \caption{Descriptions of notations used in the proof.}\label{tab:proofnotation}
\end{table}
\footnotetext{The nodes in $T_d(v)$ form a subset of $T_d$, but not necessarily a subtree. See details later.}

We start with some notation, which is summarized in
Table~\ref{tab:proofnotation}.  In the three algorithms the first tree ($T_1$) is split by the keys in the second tree ($T_2$).  We
therefore call the first tree the \emph{decomposed tree} and the second the
\emph{pivot tree}, denoted as $T_d$ and $T_p$ respectively.  The
tree that is returned is denoted as $T_r$.
Since our proof works for
either tree being larger, we use $m=\min(|T_p|, |T_d|)$ and
$n=\max(|T_p|,|T_d|)$.
We denote the subtree rooted at $v \in T_p$ as $T_p(v)$, and the tree
of keys from $T_d$ that $v$ splits as $T_d(v)$ (i.e.,
\split$(v,T_d(v))$ is called at some point in the algorithm).  For $v
\in T_p$, we refer to $|T_d(v)|$ as its \emph{splitting size}.

Figure \ref{fig:workproof} (a) illustrates the pivot tree with the
splitting size annotated on each node.  Since \split{} has logarithmic
work, we have,
\[\mbox{split work } = O\left(\sum_{v\in T_p} (\log|T_d(v)|+1)\right),\]
which we analyze in Theorem \ref{thm:workofsplit}.  We first, however, show that the join work is
bounded by the split work.  We use the following Lemma, which is
proven in the appendix.

\begin{lemma}
\label{lem:AVLRBranksum}
For $T_r=$\union$(T_p,T_d)$ on AVL, RB or WB trees,
if $r(T_p)> r(T_d)$ then $r(T_r) \le r(T_p)+r(T_d)$.
\end{lemma}

\begin{theorem}
\label{thm:joinlarger}
For each function call to \union{}, \intersect{} or \difference{} on
trees $T_p(v)$ and $T_d(v)$, the work to do the \join{} (or \joinTwo{})
is asymptotically no more than the work to do the \split{}.
\end{theorem}
\begin{proof}
For \intersect{} or \difference{}, the cost of \join{} (or \joinTwo{})
is $O(\log (|T_r|+1))$. Notice that \difference{} returns the keys in $T_d\backslash T_p$.
Thus for both \intersect{} and \difference{} we have $T_r\subseteq T_d$. The join work is
$O(\log (|T_r|+1))$, which is no more than $O(\log (|T_d|+1))$ (the split work).

For \union{}, if $r(T_p)\le r(T_d)$, the \join{} will cost $O(r(T_d))$, which is no more than the split work.

Consider $r(T_p)>r(T_d)$ for AVL, RB or WB trees. The rank of $L(T_p)$ and $R(T_p)$,
which are used in the recursive calls, are at least $r(T_p)-1$.
Using Lemma \ref{lem:AVLRBranksum}, the rank of the two trees
returned by the two recursive calls will be at least $(r(T_p)-1)$ and
at most $(r(T_p)+r(T_d))$, and differ by at most $O(r(T_d))=O(\log
|T_d|+1)$.  Thus the join cost is $O(\log|T_d|+1)$, which is no more than
the split work.

Consider $r(T_p)>r(T_d)$ for treaps. If $r(T_p)> r(T_d)$, then $|T_p|\ge |T_d|$.
The root of $T_p$ has the highest priority among all $|T_p|$ keys, so on expectation it takes at most the $\frac{|T_p|+|T_d|}{|T_p|}\le 2$-th place among all the $|T_d|+|T_p|$ keys.
From Lemma \ref{lem:treap} we know that the cost on expectation is $\mathbb{E}[\log t]+1\le \log \mathbb{E}[t]+1\le \log 2+1$, which is a constant.
\end{proof}

\begin{figure}%[!h!t]
  \centering
  \includegraphics[width=0.6\columnwidth]{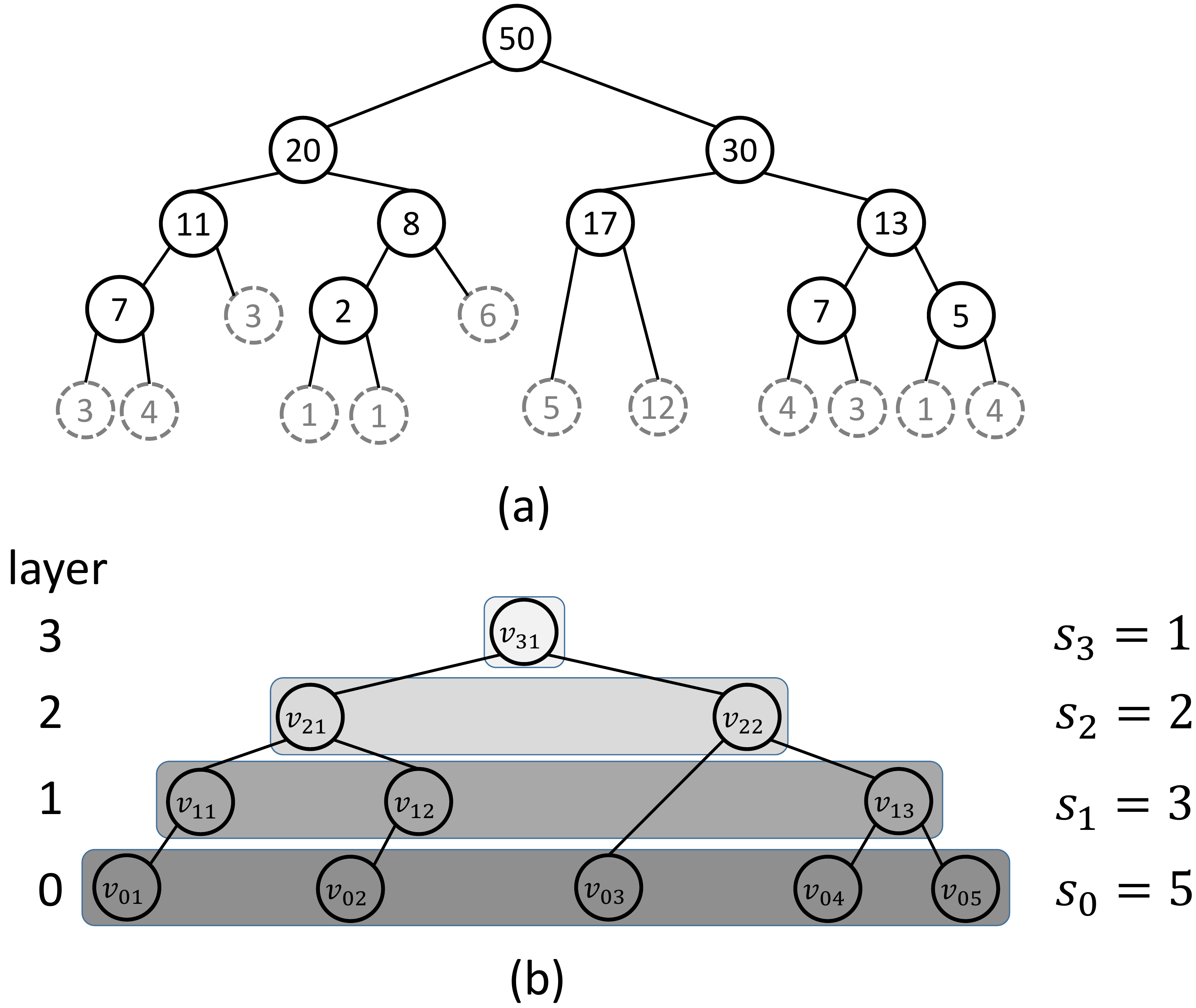}\\
  \caption{An illustration of splitting tree and layers. The tree in (a) is $T_p$, the dashed circle are the exterior nodes. The numbers on the nodes are the sizes of the tree from $T_d$ to be split by this node, i.e., the ``splitting size'' $|T_d(v)|$. In (b) is an illustration of layers on an AVL tree.}\label{fig:workproof}
\end{figure}

This implies the total join work is asymptotically bounded by the
split work.

We now analyze the split work.  We do this by layering the pivot tree
starting at the leaves and going to the root and such that nodes in a
layer are not ancestors of each other.  We define layers based on the
ranks and denote the size of layer $i$ as $s_i$.  We show that $s_i$
shrinks geometrically, which helps us prove our bound on the split
work.  For AVL and RB trees, we group the nodes with rank $i$ in layer $i$. For WB trees and treaps, we put a node $v$ in layer $i$ iff $v$ has rank $i$ and $v$'s parent has rank strictly greater than $i$. Figure \ref{fig:workproof} (b) shows an example of the layers
of an AVL tree.

\begin{definition}
In a BST, a set of nodes $V$ is called a \emph{disjoint set} if and only if for any two nodes $v_1, v_2$ in $V$, $v_1$ is not the ancestor of $v_2$.
\end{definition}

\begin{lemma}
\label{lem:lessthanm}
\hspace{-5pt}
For any disjoint set $V \subseteq T_p$, $\sum_{v\in V} |T_d(v)|\le |T_d|$.
\end{lemma}
The proof of this Lemma is straightforward.

\begin{lemma}
\label{lem:shrink}
For an AVL, RB, WB tree or a treap of size $N$, each layer is a disjoint set, and $s_i\le
\frac{N}{ c^{\lfloor i/2\rfloor}}$ holds for some constant $c>1$.
\end{lemma}
\begin{proof}
For AVL, RB, WB trees and treaps, a layer is obviously a disjoint set: a node and its ancestor cannot lie in the same layer.

For AVL trees, consider a node in layer 2, it must have at least two descendants in layer 0. Thus $s_0 \ge  2s_2$.
Since an AVL tree with its leaves removed is still an
AVL tree, we have $s_i\ge 2s_{i+2}$.  Since $s_0$ and $s_1$ are no more than $N$, we can get that
$s_i<\frac{N}{2^{\lfloor i/2 \rfloor}}$.

For RB trees, the number of black nodes in layer $2i$ is more than twice as many as in layer $2(i+1)$ and less than four times as many as in layer $2(i+1)$,
i.e., $s_{2i} \ge 2s_{2i+2}$. Also, the number of red nodes in layer $2i+1$ is no more than the black nodes in layer $2i$. Since $s_0$ and $s_1$ are no more than $N$, $s_i<\frac{N}{2^{\lfloor i/2 \rfloor}}$.

For WB trees and treaps, the rank is defined as $\lceil \log_2(w(T)) \rceil -1$, which means that a node in layer $i$ has weight at least $2^{i}+1$. Thus $s_i \le (N+1)/(2^i+1)\le N/2^i$.
\end{proof}

Not all nodes in a WB tree or a treap are assigned to a layer. We call a node a \emph{layer root} if it is in a layer. We attach each node $u$ in the tree to the layer root that is $u$'s ancestor and has the same rank as $u$. We denote $d(v)$ as the number of descendants attached to a layer root $v$.

\begin{lemma}
\label{lem:treapchain}
For WB trees and treaps, if $v$ is a layer root, $d(v)$ is less than a constant (in expectation for treaps). Furthermore, the random variables $d(v)$ for all layer roots in a treap are i.i.d. (See the proof in the Appendix.)
\end{lemma}

By applying Lemma \ref{lem:shrink} and \ref{lem:treapchain} we prove the split work. In the following proof, we denote $v_{kj}$ as the $j$-th node in layer $k$.
\begin{theorem}
\label{thm:workofsplit}
The split work in \union{}, \intersect{} and \difference{} on two trees of size $m$ and $n$ is $O\left(\boundcontent\right)$.% if one tree always serves as the pivot tree.
\end{theorem}
\begin{proof}
The total work of \split{} is the sum of the log of all the splitting
sizes on the pivot tree $O\left(\sum_{v\in T_p}\log
(|T_d(v)|+1)\right)$. Denote $l$ as the number of layers in the
tree. Also, notice that in the pivot tree, in each layer there are at most $|T_d|$ nodes with $|T_d(v_{kj})|>0$.
Since those nodes with splitting sizes of $0$ will not cost any work,
we can assume $s_i\le |T_d|$. We calculate the dominant term
$\sum_{v\in T_p}\log (|T_d(v)|+1)$ in the complexity by summing the
work across layers:
\begin{align*}
%\label{workequation}
\sum_{k=0}^{l} \sum_{j=1}^{s_k} \log \left(|T_d(v_{kj})|+1\right) &\le  \sum_{k=0}^{l} s_k \log \left(\frac{\sum_j |T_d(v_{kj})|+1}{s_k}\right)\\
&=\sum_{k=0}^{l} s_k\log\left(\frac{|T_d|}{s_k}+1\right)
\end{align*}

We split it into two cases. If $|T_d|\ge |T_p|$, $\frac{|T_d|}{s_k}$ always dominates $1$. we have:
\begin{align}
\label{eqn:from1}
\sum_{k=0}^{l} s_k \log \left(\frac{|T_d|}{s_k}+1\right)=~&\sum_{k=0}^{l} s_k \log \left(\frac{n}{s_k}+1\right)\\
\label{eqn:to1}
\le ~&\sum_{k=0}^{l} \frac{m}{c^{\lfloor k/2 \rfloor}} \log \left(\frac{n}{m/c^{\lfloor k/2 \rfloor}}+1\right)\\
\nonumber
\le ~&2\sum_{k=0}^{l/2} \frac{m}{c^k} \log \frac{n}{m/c^k}\\
\nonumber
\le~& 2\sum_{k=0}^{l/2} \frac{m}{c^k} \log \frac{n}{m}+2\sum_{k=0}^{l/2} k\frac{m}{c^k}\\
\nonumber
=~&O\left(m\log\frac{n}{m}\right)+O(m)\\
\label{eq:mlargern}
=~&O\left(\boundcontent\right)
\end{align}

If $|T_d|< |T_p|$, $\frac{|T_d|}{s_k}$ can be less than $1$ when $k$ is smaller, thus the sum should be divided into two parts. Also note that we only sum over the nodes with splitting size larger than $0$. Even though there could be more than $|T_d|$ nodes in one layer in $T_p$, only $|T_d|$ of them should count. Thus we assume $s_k\le |T_d|$, and we have:
\begin{align}
\label{eqn:from2}
\sum_{k=0}^{l} s_k \log \left(\frac{|T_d|}{s_k}+1\right)=~&\sum_{k=0}^{l} s_k \log \left(\frac{m}{s_k}+1\right)\\
\nonumber
\le ~& \sum_{k=0}^{2\log_c \frac{n}{m}} |T_d| \log \left(1+1\right) \\
\label{eqn:to2}
+& \sum_{k=2\log_c \frac{n}{m}}^{l} \frac{n}{c^{\lfloor k/2 \rfloor}} \log \left(\frac{m}{n/c^{\lfloor k/2 \rfloor}}+1\right)\\
\nonumber
%\le~& O(n\log \frac{m}{n}) + 2\sum_{k=\log_c \frac{m}{n}}^{l/2} \frac{m}{c^k} \log (\frac{n}{m/c^k})\\
%\nonumber
\le~& O\left(m\log \frac{n}{m}\right) + 2\sum_{k'=0}^{\frac{l}{2}-\log_c \frac{m}{n}} \frac{m}{c^{k'}} \log c^{k'}\\
\nonumber
=~&O\left(m\log\frac{n}{m}\right)+O(m)\\
\label{eq:nlargerm}
=~&O\left(m\log(\frac{n}{m}+1)\right)
\end{align}

From (\ref{eqn:from1}) to (\ref{eqn:to1}) and (\ref{eqn:from2}) to (\ref{eqn:to2}) we apply Lemma \ref{lem:shrink} and the fact that $f(x)=x\log ({n\over x}+1)$ is monotonically increasing when $x\le n$.

For WB trees and treaps, the calculation above only includes the log of splitting size on layer roots. We need to further prove the total sum of the log of all splitting size is still $O\left(\boundcontent\right)$. Applying Lemma \ref{lem:treapchain}, the expectation is less than:
\begin{align*}
&\mathbb{E}\left[2\sum_{k=0}^{l} \sum_{j=1}^{x_k} d(v_{kj})\log ((T_d(v_{kj})+1)\right]\\
=~&\mathbb{E}[d(v_{kj})]\times2\sum_{k=0}^{l} \sum_{j=1}^{x_k} \log ((T_d(v_{kj})+1)\\
=~&O\left(\boundcontent\right)
\end{align*}
For WB trees $d(v_{kj})$ is no more than a constant, so we can also come to the same bound.
% that $s(v_{kj})$ are i.i.d. with expected value of a constant.

To conclude, the split work on all four balancing schemes of all three functions is $O\left(\boundcontent\right)$.
\end{proof}

\begin{theorem}
\label{thm:mainwork}
The total work of \union{}, \intersect{} or \difference{} of all four balancing schemes on two trees of size $m$ and $n$ ($m\ge n$) is $O\left(\boundcontent\right)$.
\end{theorem}
This directly follows Theorem \ref{thm:joinlarger} and \ref{thm:workofsplit}.

\begin{theorem}
\label{thm:mainspan}
The span of \union{} and \intersect{} or \difference{} on all four balancing schemes is $O(\log n \log m$). Here $n$ and $m$ are the sizes of the two tree.
\end{theorem}
\begin{proof}
For the span of these algorithms, we denote $D(h_1,h_2)$ as the span on \union{}, \intersect{} or \difference{} on two trees of height $h_1$ and $h_2$. According to Theorem \ref{thm:joinlarger}, the work (span) of \split{} and \join{} are both $O(\log |T_d|)=O(h(T_d))$. We have:
$$D(h(T_p),h(T_d))\le D(h(T_p)-1,h(T_d))+2h(T_d)$$
Thus $D(h(T_p),h(T_d))\le 2h(T_p)h(T_d)=O(\log n \log m)$.
\end{proof}

Combine Theorem \ref{thm:mainwork} and \ref{thm:mainspan} we come to Theorem \ref{thm:unionwork}.

\section{Experiments}
\label{sec:exp}
To evaluate the performance of our algorithms we performed several
experiments across the four balancing schemes using different set functions,
while varying the core count and tree sizes.
We also compare the performance of our implementation to other existing libraries and algorithms.
%For the sequential setting we compare ourself to the STL's \texttt{std::set} and \texttt{std::vector},
%while for parallel experiments we compare our code to the MCSTL (Multi Core Standard Template) library \cite{FS07} and
%parallel weight-balanced B-trees (WBB-tree) \cite{WBTree}.

\hide{
We compare the performance of our library on a single core to the C++
Standard Template Library (STL) implementation of ordered maps, which is
based on red-black trees, and their generic implementation of
\union{}.
}
\paragraph{Experiment setups and baseline algorithms}
For the experiments we use a 64-core machine with 4 x AMD Opteron(tm)
Processor 6278 (16 cores, 2.4GHz, 1600MHz bus and 16MB L3 cache). Our code was
compiled using the \texttt{g++} 4.8 compiler with the Cilk Plus
extensions for nested parallelism. The only compilation flag we used was the \texttt{-O2} optimization flag. In all our experiments we use keys
of the double data type. The size of each node is about 40 bytes,
including the two child pointers, the key, the balance
information, the size of the subtree, and a reference count. We
generate multiple sets varying in size from $10^4$ to $10^8$. Depending on the experiment the keys are drawn either from an uniform or a Gaussian distribution.
We use $\mu$ and $\sigma$ to denote the mean and the standard deviation in Gaussian distribution.
Throughout this section $n$ and $m$ represent the two input sizes for functions with two input sets ($n\ge m$).

We test our algorithm by comparing it to other available
implementations. This includes the sequential version of the set
functions defined in the C++ Standard Template Library
(STL)~\cite{musser2009stl} and STL's \texttt{std::set} (implemented by
RB tree).  The STL supports the set operations \texttt{set\_union},
\texttt{set\_intersection}, and \texttt{set\_difference} on any
container class. Using an \texttt{std::vector} these algorithms
achieve a runtime of $O(m + n)$.  Since the STL does not offer any
parallel version of these functions we could only use it for
sequential experiments.  To see how well our algorithm performs in a
parallel setting, we compare it to parallel WBB-trees \cite{WBTree}
and the MCSTL library \cite{FS07}.  WBB-trees, as well as the MCSTL,
offer support for bulk insertions and deletions. They process the bulk
updates differently.  The MCSTL first splits the main tree
among $p$ processors, based on the bulk sequence, and then inserts the
chunks dynamically into each subtree.  The WBB-tree recursively inserts
the bulk in parallel into the main tree.  To deal with heavily
skewed sequences they use partial tree reconstruction for fixing
imbalances, which takes constant amortized time. The WBB-tree has a
more cache aware layout, leading to a better cache utilization
compared to both the MCSTL and our implementation. To make a
comparison with these implementations we use their bulk insertions,
which can also be viewed as an union of two sets. However notice that
WBB-trees take the bulk in the form of a sorted sequence, which gives
them an advantage due to the faster access to the one array than to a tree, and
far better cache performance (8 keys per cache line as opposed to 1).

\begin{center}
\begin{figure*}[t!]
\begin{minipage}{.24\textwidth}
  \centering \includegraphics[width=1\columnwidth]{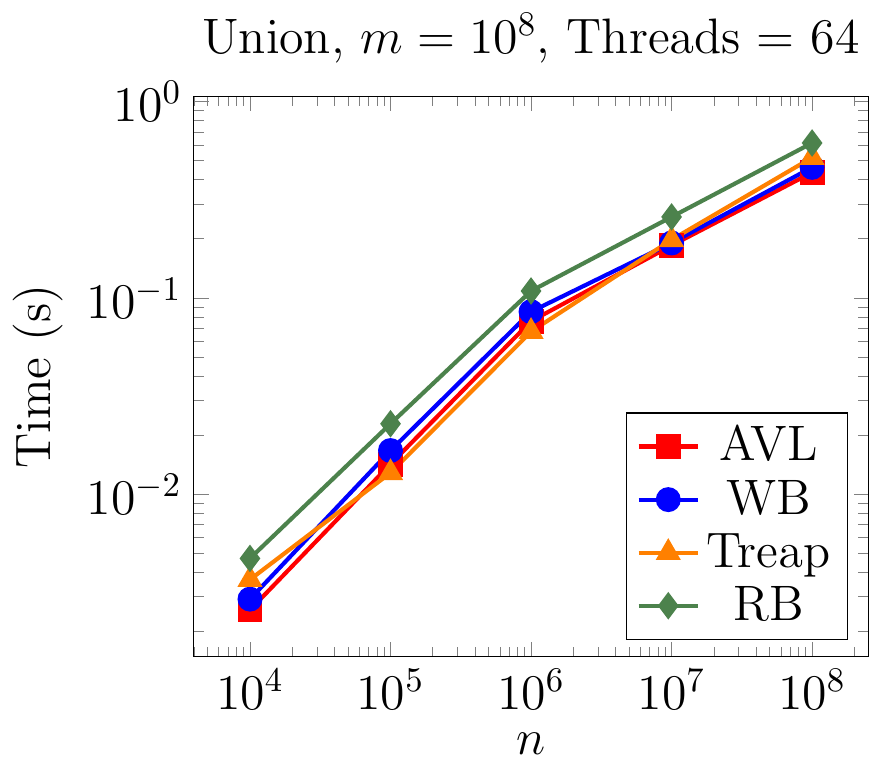} \\ (a)
\end{minipage}
\begin{minipage}{.24\textwidth}
\centering \includegraphics[width=1\columnwidth]{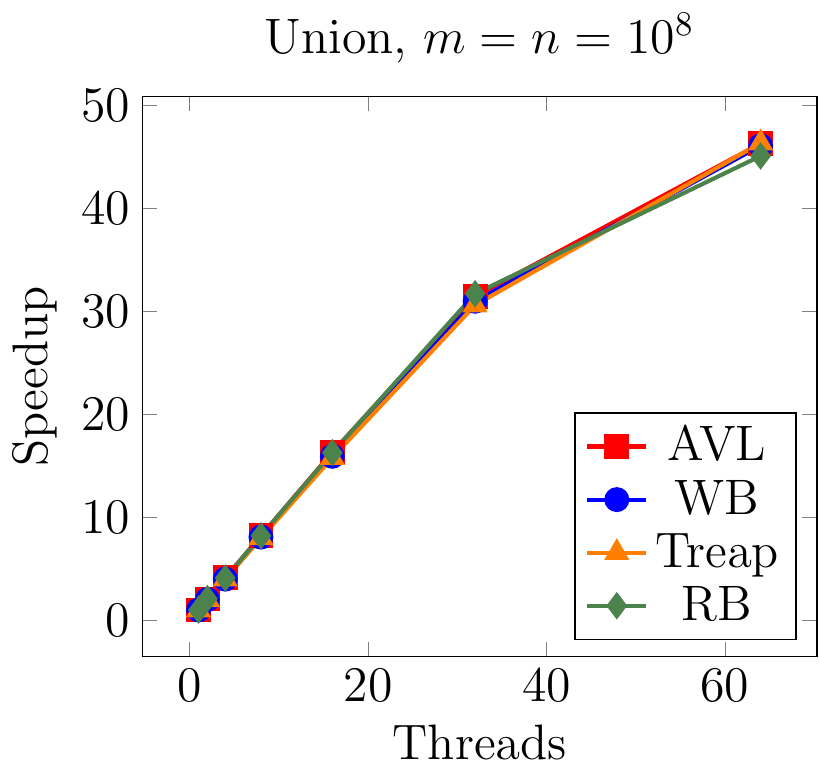} \\ (b)
\end{minipage}
\begin{minipage}{.24\textwidth}
\centering \includegraphics[width=1\columnwidth]{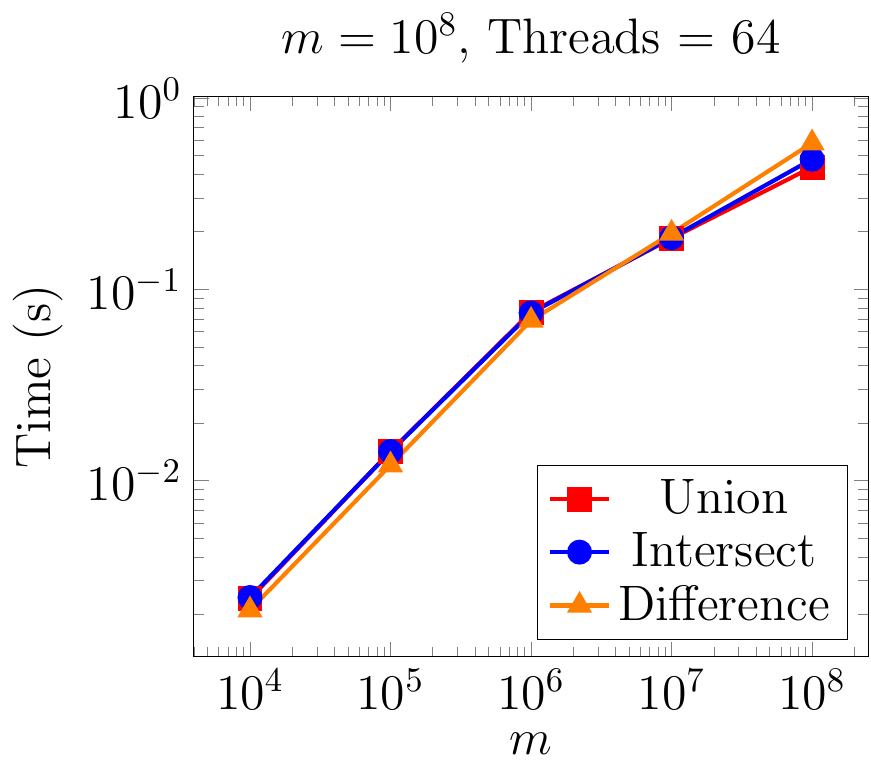} \\ (c)
\end{minipage}
%\vspace{.1in}
\begin{minipage}{.24\textwidth}
  \centering \includegraphics[width=1\columnwidth]{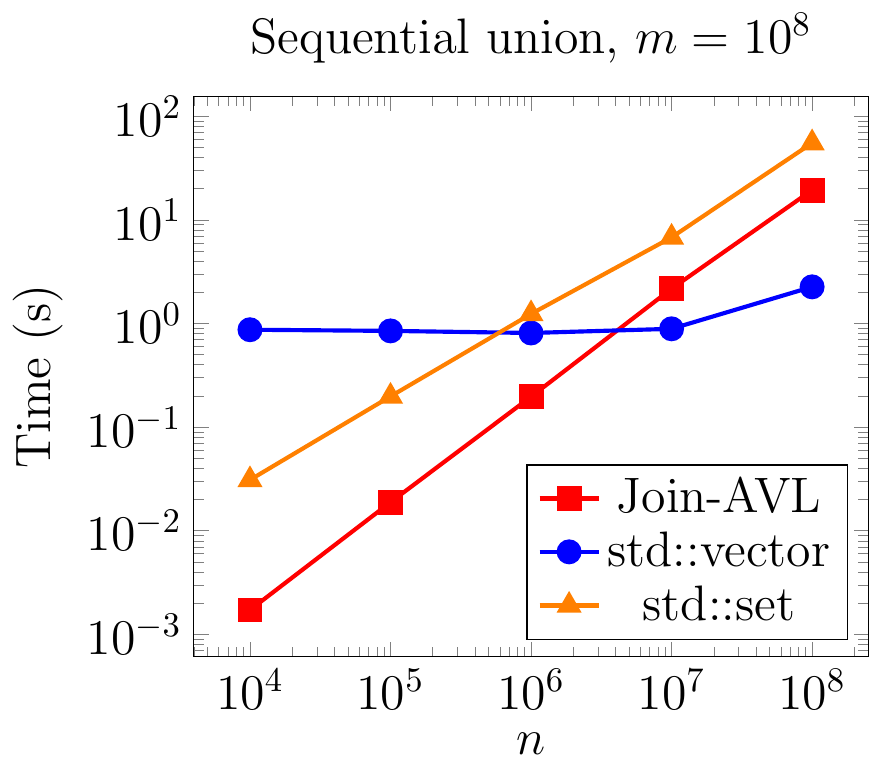} \\ (d)
\end{minipage}

\vspace{.1in}

\begin{minipage}{.24\textwidth}
\centering\includegraphics[width=1\columnwidth]{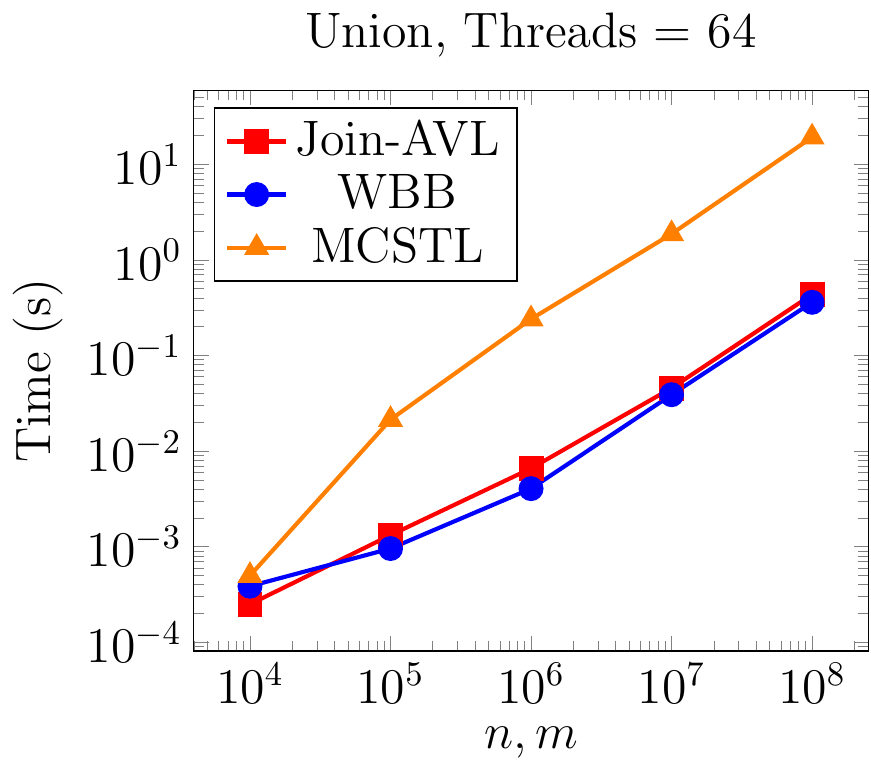}\\(e)
\end{minipage}
%\begin{minipage}{.32\textwidth}
%  \centering \includegraphics[scale=\figscale]{figures/new-graphs/uniform_string} \\ (f)
%\end{minipage}
%\vspace{.1in}
\begin{minipage}{.24\textwidth}
\centering\includegraphics[width=1\columnwidth]{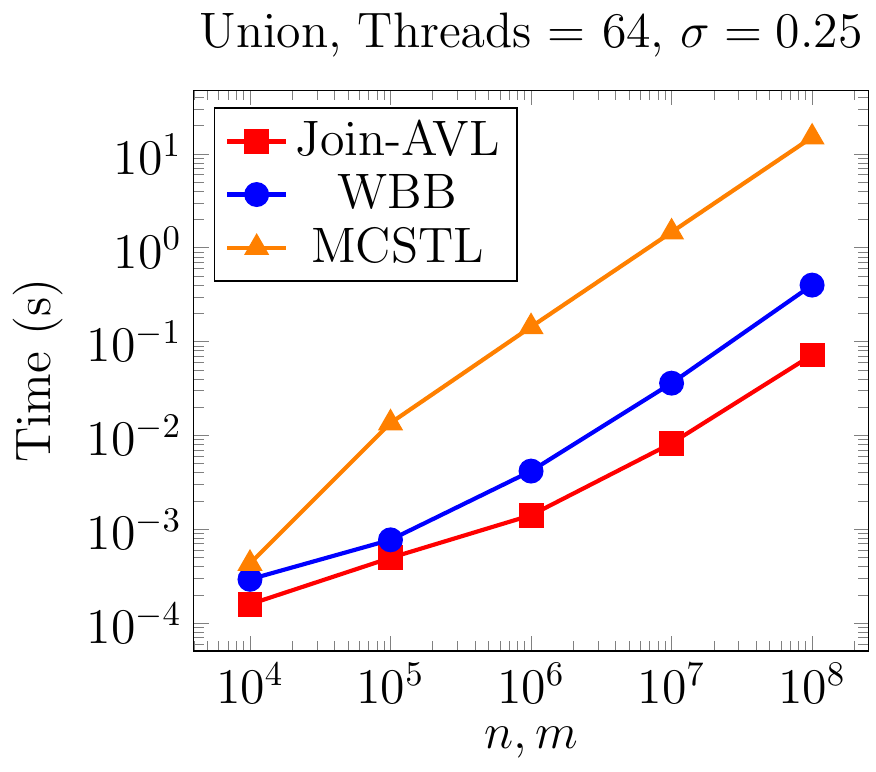}\\(f)
\end{minipage}
\begin{minipage}{.24\textwidth}
\centering\includegraphics[width=1\columnwidth]{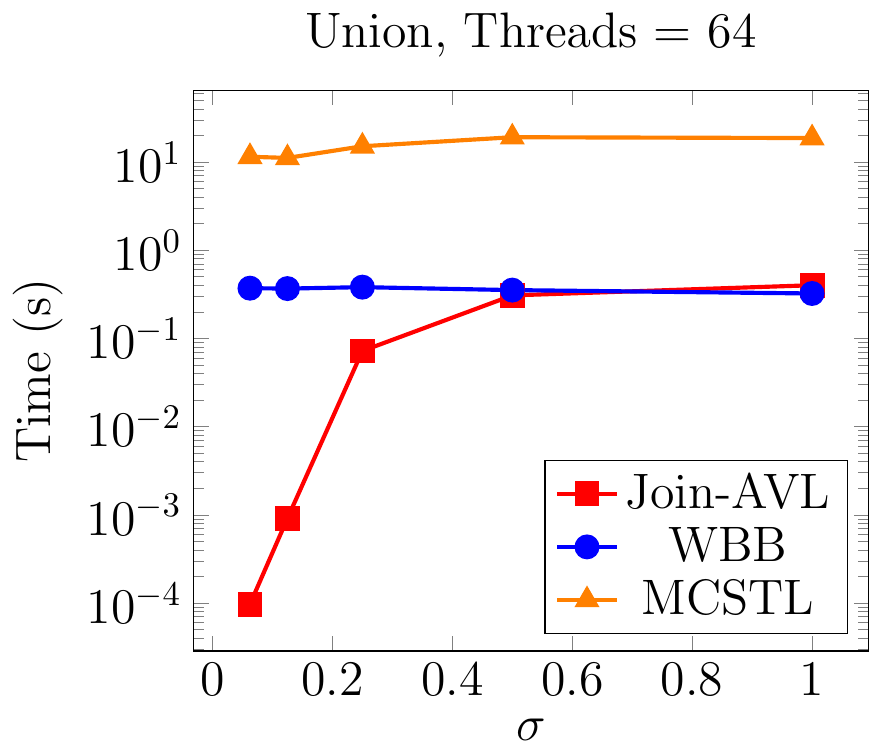}\\(g)
\end{minipage}
\begin{minipage}{.24\textwidth}
\centering\includegraphics[width=1\columnwidth]{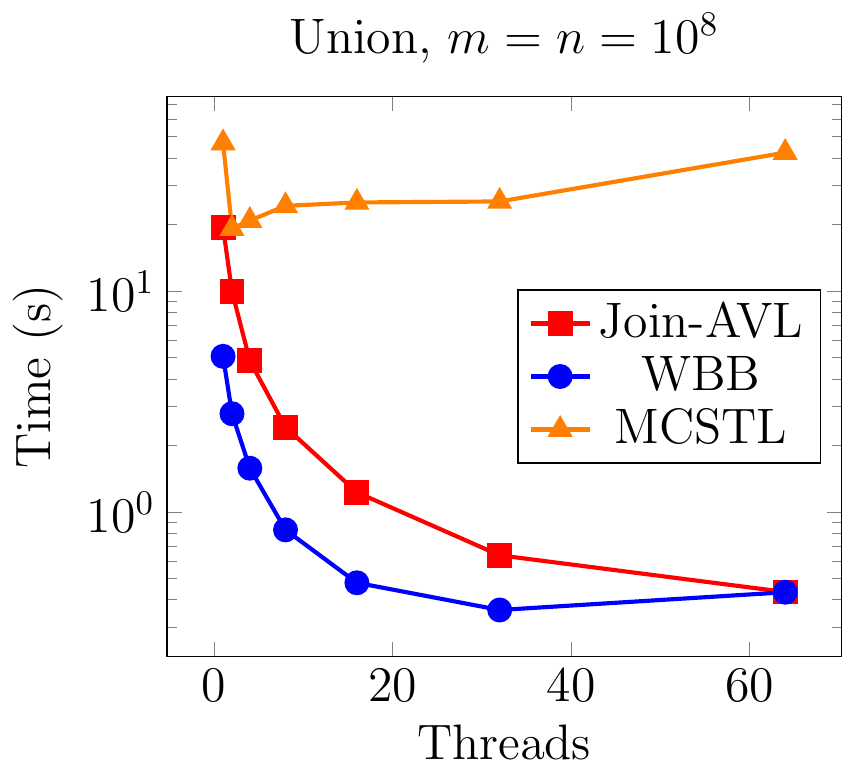}\\(h)
\end{minipage}

%\nocaptionrule
\caption{
(a) Times for \union{} as a function of size ($n = 10^8$) for different
BBSTs;
(b) speed up of \union{} for different BBSTs;
(c) times for various operations on AVL trees as a function of size ($n  = 10^8$);
(d) comparing STLs \texttt{set\_union} with our \union{};
(e, f, g, h) comparing our \union{} to other parallel search trees; (e, h) input keys are uniformly distributed doubles in the range of [0, 1];
(f, g) inputs keys follow a normal distribution of doubles -
the mean of the main tree is always $\mu_1=0$, while the mean of the bulk is $\mu_2=1$. Figure (f) uses a standard deviation of $\sigma=0.25$, while
Figure (g) shows the performance across different standard deviations.
}
\label{fig:times}
\end{figure*}
\end{center}

\paragraph{Comparing the balancing schemes and functions}
To compare the four balancing schemes we choose \union{} as the
representative operation. Other operations would lead to similar
results since all operations except \join{} are generic across the
trees. We compare the schemes across different thread counts and different sizes.

Figure~\ref{fig:times} (a) shows the runtime of \union{} for various
tree sizes and all four balancing schemes across 64
cores. The
times are very similar among the balancing schemes---they
differ by no more than 10\%.

Figure~\ref{fig:times} (b) shows the speedup curves for \union{} on varying
core numbers
with two inputs of size $10^8$. All balancing
schemes achieve a speedup of about 45 on 64 cores, and about 30 on 32
cores. The less-than-linear speedup beyond 32 cores is not
due to lack of parallelism, since when we ran the same experiments on
significantly smaller input (and hence less parallelism) we get very
similar curves (not shown). Instead it seems to be due to saturation
of the memory bandwidth.

We use the AVL tree as the representative tree to compare different
functions. Figure~\ref{fig:times} (c) compares the \union{},
\intersect{} and \difference{} functions.  The size of the
larger tree is fixed ($10^8$), while the size of the smaller
tree varies from $10^4$ to $10^8$.
As the plot indicates, the three functions have very similar performance.

The experiments are a good indication of the performance
of different balancing schemes and different functions, while controlling other factors. The
conclusion is that all schemes perform almost equally on all the set functions. It is
not surprising that all balancing schemes achieve similar performance because the dominant cost is in cache misses along the
paths in the tree, and all schemes keep the trees reasonably balanced.
The AVL tree is always slightly faster than the other trees and this
is likely due to the fact that they maintain a slightly stricter balance than
the other trees, and hence the paths that need to be traversed are
slightly shorter. For different set functions the performance is also as expected given the similarity of the code.

Given the result that the four balancing schemes do not have a big difference in timing
and speedup, nor do the three set functions, in the following experiments we use the AVL tree along with
\union{} to make comparisons with other implementations.

\paragraph{Comparing to sequential implementations}
The STL supports \texttt{set\_union} on any sorted container class, including sets
based on red-black trees, and sorted vectors (arrays). The STL
\texttt{set\_union} merges the two sorted containers by moving from left to right
on the two inputs, comparing the current values, and writing the lesser
to the end of the output. For two inputs of size $n$ and $m$, $m \le
n$, it takes $O(m + n)$ time on \texttt{std::vectors}, and $O((n+m) \log (n+m))$
time on \texttt{std::set} (each insertion at the end of the output
red-black tree takes $O(\log(n+m))$ time). In the case of ordered sets we can do better
by inserting elements from the smaller set into the larger, leading a time of $O(m \log (n+m)$.
This is also what we do in our experiments. For vectors we stick with the available \texttt{set\_union} implementation.

  Figure~\ref{fig:times} (d)
gives a comparison of times for \union{}. For equal lengths our implementation is
about a factor of 3 faster than set variant (red-black trees), and about
8 times slower than the vector variant. This is not surprising since we are
asymptotically faster than their red-black tree implementation, and
their array-based implementation just reads and writes the values, one by one,
from flat arrays, and therefore has much less overhead and much
fewer cache misses. For taking the union of smaller and larger inputs,
our \union{} is orders of magnitude faster than either STL version.
This is because their theoretical work bound ($O(m+n)$ and $O(m\log(m+n)$)
is worse than our $O(m \log (n/m+1))$, which
is optimal in comparison model.
%they both take time proportional to the larger length
%while ours takes $O(n \log (m/n))$  work (time).

\paragraph{Comparing to parallel implementations on various input distributions}
We compare our implementations to other parallel search trees, such as the WBB-trees, as described in \cite{WBTree}, and parallel
RB trees from the MCSTL~\cite{FS07}. We test the performance on different input distributions.

In Figure \ref{fig:times} (e) we show the result of \union{} on
uniformly distributed doubles in the range of [0,1] across 64
cores. We set the input size to $n=m=10^i$, $i$ from $4$ to $8$. The
three implementations have similar performance when $n=m=10^4$. As the
input size increases, MCSTL shows much worse performance than the
other two because of the lack of parallelism (Figure \ref{fig:times}
(h) is a good indication), and the WBB-tree implementation is slightly
better than ours.  For the same reason that STL vectors outperform STL
sets (implemented with RB trees) and our sequential implementation,
the WBB-trees take the bulk as a sorted array, which has much less
overhead to access and much better cache locality.   Also their tree
layout is more cache efficient and the overall height is lower since
they store multiple keys per node.

%Figure \ref{fig:times} (f) shows the result on strings on all 64 cores with set sizes $n=m=10^i$ for $i=4$ through $8$. The strings are of length from 3 through 20, each character
%picked randomly from a to z. The trend of the curves in (f) are very similar to (e).

Figure \ref{fig:times} (f) shows the result of a Gaussian distribution with doubles, also on all 64 cores with set sizes of $n=m=10^i$ for $i=4$ through $8$. The distributions of the two sets have
means at $0$ and $1$ respectively, and both having a standard deviation of $0.25$, meaning that the data in the two sets have less overlap comparing to a uniform distribution (as in (e)). In this case our code achieves better performance than the other two implementations. For our algorithms less overlap in the data means more parts of the trees will be untouched, and therefore less nodes will be operated on. This in turn leads to less processing time.

We also do an in-depth study on how the overlap of the data sets affects the performance of each algorithm. We generate two sets of size $n=m=10^8$, each from a Gaussian distribution.
The distributions of the two sets have means at $0$ and $1$ respectively, and both have an equal standard deviation varying in $\{1,1/2,1/4,1/8,1/16\}$. The different standard deviations are to control
the overlap of the two sets, and ideally less overlap should simplify the problem. Figure \ref{fig:times} (g) shows the result of the three parallel implementations on a Gaussian distribution with different standard deviations. From the figure we can see that MCSTL and WBB-tree are not affected by different standard deviations, while our join-based union takes advantage of less overlapping and achieves a much better performance when $\sigma$ is small. This is not surprising since when the two sets are less overlapped, e.g., totally disjoint, our \union{} will degenerate to a simple \join{}, which costs only $O(\log n)$ work. This behavior is consistent with the ``adaptive'' property (not always the worst-case) in \cite{demaine2000adaptive}. This indicates that our algorithm is the only one among the three parallel implementations that can detect and take advantage of less overlapping in data, hence have a much better performance when the two operated sets are less overlapped.

We also compare the parallelism of these implementations. In Figure \ref{fig:times} (h) we show their performance across 64 cores.
The inputs are both of size $10^8$, and generated from an uniform distribution of doubles.
It is easy to see that MCSTL does not achieve good parallelism beyond 16 cores, which explains why the MCSTL always performs the worst on 64 cores in all settings.
As we mentioned earlier, the WBB-tree are slightly faster than our code, but when it comes to all 64 cores, both algorithms have similar performance. This indicates that our algorithm achieves better parallelism.

To conclude, in terms of parallel performance, our code and WBB-trees are always much better than the MCSTL because of MCSTL's lack of parallelism. WBB-trees achieve a slightly better performance than ours on uniformly distributed data, but it does not improve when the two sets are less overlapped. Thus our code is much better than the other two implementations on less overlapped data, while still achieving a similar performance with the other algorithms when the two sets are more intermixed with each other.

\newpage
\section{Conclusions}
\label{conclusion}
In this paper, we study ordered sets implemented with balanced binary search trees.  We show for the first time that a very simple ``classroom-ready'' set of algorithms due to Adams' are indeed work optimal when used with four different balancing schemes--AVL, RB, WB trees and treaps---and also highly parallel.  The only tree-specific algorithm that is necessary is the \join, and even the \join{}s are quite simple, as simple as \insertnew{} or \delete{}.   It seems it is not sufficient to give a time bound to \join{} and base analysis on it.   Indeed if this were the case it would have been done years ago.   Instead our approach defines the notion of a rank (differently for different trees) and shows invariants on the rank.  It is important that the cost of \join{} is proportional to the difference in ranks.     It is also important that when joining two trees the resulting rank is only a constant bigger than the larger rank of the inputs.    This insures that when joins are used in a recursive tree, as in \union{}, the ranks of the results in a pair of recursive calls does not differ much on the two sides.  This then ensures that the set functions are efficient.

We also test the performance of our algorithm. Our experiments show that our sequential algorithm is about 3x faster for union on two maps of size $10^8$ compared to the STL red-black tree implementation.  In parallel settings our code is much better than the two baseline algorithms (MCSTL and WBB-tree) on less overlapped data, while still achieves similar performances with WBB-tree when the two sets are more intermixed. Our code also achieves 45x speedup on 64 cores.

\bibliographystyle{abbrv}
\bibliography{main}
\appendix
\section{Proofs for Some Lemmas}

\subsection{Proof for Lemma \ref{lem:treapchain}}
\label{appendix:treapchain}
\begin{proof}
One observation in WB trees and treaps is that all nodes attached to a
single layer root form a chain. This is true because if two children of one node $v$ are both in layer $i$, the weight of $v$ is more than $2^{i+1}$, meaning that $v$ should be layer $i+1$.

For a layer root $v$ in a WB tree on layer $k$, $w(v)$ is at most $2^{k+1}$. Considering the balance invariant that its child has weight at most $(1-\alpha)w(v)$, the weight of the $t$-th generation of its descendants is no more than $2^{k+1}(1-\alpha)^{t}$. This means that after $t^{*}=\log_{\frac{1}{1-\alpha}}2$ generations, the weight should decrease to less than $2^k$. Thus $d(v)\le \log_{\frac{1}{1-\alpha}}2$, which is a constant.

For treaps consider a layer root $v$ on layer $k$ that has weight
$N\in[2^k,2^{k+1})$.  The probability that $d(v)\ge 2$ is equal to the
probability that one of its grandchildren has weight at least $2^k$.  This
probability $P$ is:

\begin{align}
P&=\frac{1}{2^k}\sum_{i=2^k+1}^{N}\frac{i-2^k}{i}\\
&\le\frac{1}{2^k}\sum_{i=2^k+1}^{2^{k+1}}\frac{i-2^k}{i}\\
%&=1-\sum_{i=2^k+1}^{2^{k+1}}\frac{1}{i}\\
&\approx 1-\ln 2
\end{align}
We denote $1-\ln 2$ as $p_c$. Similarly, the probability that $d(v)\ge 4$ should be less than $p_c^2$, and the probability shrink geometrically as $d(v)$ increase. Thus the expected value of $d(v)$ is a constant.

Since treaps come from a random permutation, all $s(v)$ are i.i.d.
\end{proof}

\subsection{Proof for Lemma \ref{lem:AVLRBranksum}}
\label{appendix:AVLRBranksum}
\begin{proof}
  We are trying to show that for $T_r=$\union$(T_p,T_d)$ on AVL, RB or
  WB trees, if $r(T_p)> r(T_d)$ then $r(T_r) \le r(T_p)+r(T_d)$.

  For AVL and RB trees we use induction on $r(T_p)+r(T_d)$.  When
  $r(T_d)+r(T_p)=1$ the conclusion is trivial.  If $r=r(T_p)>r(T_d)$,
  $T_p$ will be split into two subtrees, with rank at most $r(T_p)-1$
  since we remove the root.  $T_d$ will be split into two trees with
  height at most $r(T_d)$ (Theorem \ref{thm:split}).
  Using the inductive hypothesis, the two recursive calls will return
  two trees of height at most $r(T_p)-1+r(T_d)$. The result of the
  final \join{} is therefore at most $r(T_p)+r(T_d)$.

For WB trees, $|T|\le |T_p|+|T_d|\le 2|T_p|$. Thus $r(T)\le r(T_p)+1 \le r(T_p)+r(T_d)$.
\end{proof}

\hide{
\subsection{Proof for Theorem \ref{thm:workofboundary}}
\label{appendix:workofboundary}
\begin{proof}
Notice that the cost of the boundary condition on each leaf is $O(1)$. Also, if the splitting size on a node in the splitting tree is $0$, it will not go further down. If the smaller tree is the pivot tree, there are at most $\min(|T_p|,|T_d|)$ leaves. Thus the total base cost is $O(\min(|T_p|,|T_d|))$. If the larger tree is the pivot tree, we remove all nodes with splitting size $0$. The remaining tree has at most $\min(|T_p|,|T_d|)$ leaves, which are the base conditions. The overall base cost is also $O(\min(|T_p|,|T_d|))$.
\end{proof}
}

\subsection{Proof for Lemma \ref{lem:wbvalid}}
\label{appendix:wbvalid}
\begin{proof}
\begin{figure}
  % Requires \usepackage{graphicx}
  \centering
  \includegraphics[width=0.6\columnwidth]{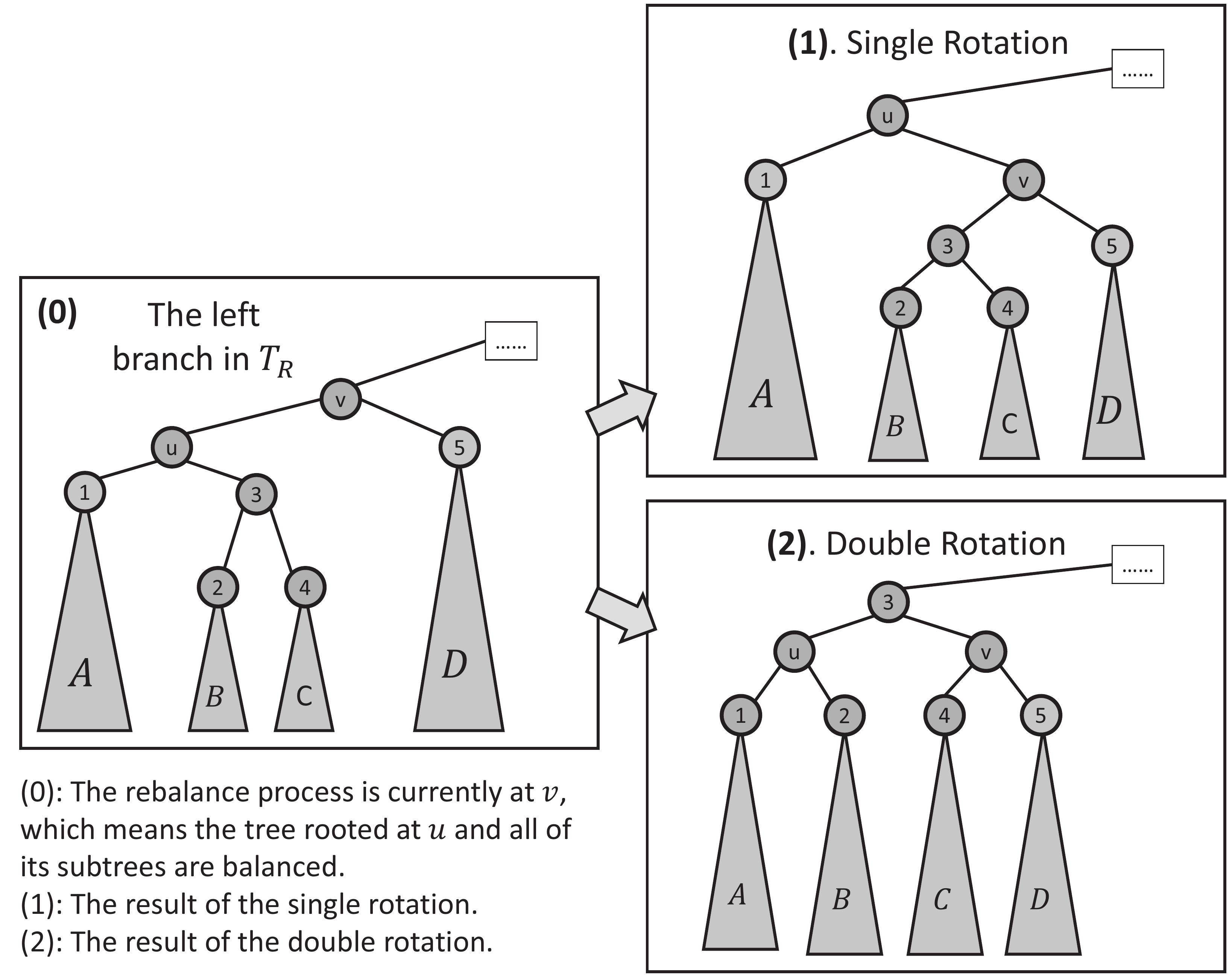}\\
  \caption{An illustration of two kinds of outcomes of rotation after joining two weight balanced trees. After we append the smaller tree to the larger one and rebalance from that point upwards, we reach the case in (0), where $u$ has been balanced, and the smaller tree has been part of it. Now we are balancing $v$, and two options are shown in (1) and (2). At least one of the two rotation will rebalance $v$.}\label{wbtree2}
\end{figure}
Recall that in a weight balanced tree, for a certain node, neither of its children is $\beta$ times larger than the other one, where $\beta=\frac{1}{\alpha}-1$. When $\alpha \le 1-\frac{1}{\sqrt{2}}$, we have $\beta\ge 1+\sqrt{2}$.

WLOG, we prove the case when $|\Tl|<|\Tr|$, where $\Tl$ is
inserted along the left branch of $\Tr$. Then we rebalance the tree
from the point of key $k$ and go upwards. As shown in Figure
\ref{wbtree2} (0), suppose the rebalance has been processed to $u$
(then we can use reduction). Thus the subtree rooted at $u$ is
balanced, and $\Tl$ is part of it. We name the four trees from left to
right $A,B,C$ and $D$, and the number of nodes in them $a,b,c$
and $d$. From the balance condition we know that $A$ is balanced with
$B+C$, and $B$ is balanced to $C$, i.e.:
\begin{align}
\label{b+ctoa}
\frac{1}{\beta}(b+c)\le&a\le\beta (b+c)\\
\label{btoc}
\frac{1}{\beta}b\le&c\le\beta c
\end{align}
We claim that at least one of the two operations will rebalanced the tree rooted at $v$ in Figure \ref{wbtree2} (0):
\begin{enumerate}[Op. (1).]
  \item Single rotation: right rotation at $u$ and $v$ (as shown in Figure \ref{wbtree2} (1));
  \item double rotation: Left rotation followed by a right rotation (as shown in Figure \ref{wbtree2} (2)).
\end{enumerate}

Also, notice that the inbalance is caused by the insertion of a
subtree at the leftmost branch. Suppose the size of the smaller tree
is $x$, and the size of the original left child of $v$ is $y$. Note
that in the process of \join{}, $T_L$ is not concatenated with
$v$. Instead, it goes down to deeper nodes.
%Guy change further to deeper?
Also, note that the original subtree of size $y$ is weight balanced with $D$. This means we have:

\begin{align*}
x&<\frac{1}{\beta}(d+y)\\
\frac{1}{\beta} d &\le y\le \beta d\\
x+y&=a+b+c
\end{align*}

From the above three inequalities we get $x< \frac{1}{\beta}d+d$, thus:
\begin{equation}
\label{a+b+csmallerd}
a+b+c=x+y<(1+\beta+\frac{1}{\beta})d
\end{equation}

Since a unbalance occurs, we have:
\begin{equation}
\label{a+b+clargerd}
a+b+c>\beta d
\end{equation}

We discuss the following 3 cases:

\begin{enumerate}[{Case} 1.]
  \item \textbf{$\bm{B+C}$ is weight balanced with $\bm D$}, i.e.,
  \begin{equation}
  \label{b+ctodbalance}
  \betaone (b+c)\le d \le \beta(b+c)
  \end{equation}

  In this case, we apply a right rotate. The new tree rooted at $u$ is now balanced. $A$ is naturally balanced. %Now we prove that $A$ is balanced with $B+C+D$.

  Then we discuss in two cases:

  \begin{enumerate}[{Case 1.}1.]
  \item $\bm{\beta a \ge b+c+d}$.

  Notice that $b+c\ge\betaone a$, meaning that $b+c+d>\betaone a$. Then in this case, $A$ is balanced to $B+C+D$, $B+C$ is balanced to $D$.
  Thus just one right rotation will rebalance the tree rooted at $u$ (Figure \ref{wbtree2} (1)).

  \item $\bm{\beta a<b+c+d}$.

  In this case, we claim that a double rotation as shown in Figure \ref{wbtree2} (2) will rebalance the tree. Now we need to prove the balance of all the subtree pairs: $A$ with $B$, $C$ with $D$, and $A+B$ with $C+D$.

  First notice that when $\beta a<b+c+d$, from (\ref{a+b+clargerd}) we can get:
  \begin{align}
  \nonumber
  &\beta d <a+b+c<\betaone(b+c+d)+b+c\\
  \nonumber
  \Rightarrow &(\beta-\betaone)d<(\betaone+1)(b+c)\\
  \label{dtob+c}
  \Rightarrow &(\beta-1)d<b+c
  \end{align}
  Considering (\ref{b+ctodbalance}), we have $(\beta-1)d<b+c\le\beta d$. Notice $b$ and $c$ satisfy (\ref{btoc}), we have:
  \begin{align}
  \label{btod}
  b>\frac{1}{\beta+1}(b+c)>\frac{\beta-1}{\beta+1}d\\
  \label{ctod}
  c>\frac{1}{\beta+1}(b+c)>\frac{\beta-1}{\beta+1}d
  \end{align}

  Also note that when $\beta>1+\sqrt{2}\approx 2.414$, we have
  \begin{equation}
  \label{beta1}
  \frac{\beta+1}{\beta-1}<\beta
  \end{equation}
  We discuss the following three conditions of subtrees' balance:
  \begin{enumerate}[I.]
    \item \textbf{Prove $\bm A$ is weight balanced to $\bm B$}.
    \begin{enumerate}[i.]
    \item \textbf{Prove} $\bm{b\le\beta a}$.

    Since $\beta a\le b+c$ (applying (\ref{b+ctoa})) , we have $b\le \beta a$.

    \item \textbf{Prove} $\bm{a\le\beta b}$.

    In the case when $\beta a<b+c+d$ we have:
    \begin{align*}
    a&<\frac{1}{\beta}(b+c+d)~~~~{\rm{(applying{(\ref{btoc}),(\ref{btod})})}}\\
    &<\betaone(b+\beta b +\frac{\beta+1}{\beta-1}b)\\
    &=\frac{\beta+1}{\beta-1}b\\
    &<\beta b
    \end{align*}
    \end{enumerate}
    \item \textbf{Prove $\bm C$ is weight balanced to $\bm D$}.
    \begin{enumerate}[i.]
    \item \textbf{Prove} $\bm{c\le \beta d}$.

    Since $b+c\le\beta d$ (applying (\ref{b+ctodbalance})), we have $c\le\beta d$.

    \item \textbf{Prove} $\bm{d\le \beta c}$.

    From (\ref{ctod}), we have
    \begin{align*}
    d<\frac{\beta+1}{\beta-1}c<\beta c
    \end{align*}
    \end{enumerate}
    \item \textbf{Prove $\bm{A+B}$ is weight balanced to $\bm{C+D}$}.
    \begin{enumerate}[i.]
    \item \textbf{Prove} $\bm{a+b\le\beta (c+d)}$.

    From (\ref{dtob+c}), (\ref{btoc}) and (\ref{beta1}) we have:
    \begin{align*}
    &d<\frac{1}{\beta-1}(b+c)\le\frac{1}{\beta-1}(\beta c+c)\\
    &=\frac{\beta+1}{\beta-1}c<\beta c\\
    \Rightarrow&\betaone d< c\\
    \Rightarrow&(1+\betaone)d<(1+\beta)c\\
    \Rightarrow&(1+\betaone+\beta)d<\beta(c+d)+c~~\apply{\ref{a+b+csmallerd}}\\
    \Rightarrow& a+b+c<(1+\betaone+\beta)d<\beta(c+d)+c\\
    \Rightarrow& a+b<\beta(c+d)
    \end{align*}
    \item \textbf{Prove} $\bm{c+d\le\beta (a+b)}$.

    When $\beta>2$, we have $\frac{\beta}{\beta-1}<\beta$. Thus applying (\ref{dtob+c}) and (\ref{b+ctoa}) we have:
    \begin{align*}
    d<\frac{1}{\beta-1}(b+c)\le\frac{\beta}{\beta-1}a<\beta a
    \end{align*}
    Also we have $c\le\beta b$ (applying (\ref{btoc})). Thus $c+d<\beta (a+b)$.
    \end{enumerate}
  \end{enumerate}
  \end{enumerate}
  \item \textbf{$\bm{B+C}$ is too light that cannot be balanced with $\bm D$}, i.e.,
  \begin{equation}
  \label{b+clessthand}
  \beta(b+c)<d
  \end{equation}

  In this case, we have $a<\beta (b+c)<d$ (applying (\ref{b+ctoa}) and (\ref{b+clessthand})), which means that $a+b+c<d+\betaone d < \beta d$ when $\beta>\frac{1+\sqrt{5}}{2}\approx 1.618$. This contradicts with the condition that $A+B+C$ is too heavy to $D$ ($a+b+c>\beta d$). Thus this case is impossible.

  \item \textbf{$\bm{B+C}$ is too heavy that cannot be balanced with $\bm D$}, i.e.,
  \begin{align}
  \label{b+clargerthand}
  &b+c>\beta d\\
  \label{alargerd}
  \Rightarrow &a>\betaone (b+c)>d
  \end{align}
  In this case, we apply the double rotation.

  We need to prove the following balance conditions:

  \begin{enumerate}[I.]
    \item \textbf{Prove $\bm A$ is weight balanced to $\bm B$}.
    \begin{enumerate}[i.]
    \item \textbf{Prove} $\bm{b<\beta a}$.

    Since $\beta a>b+c$ (applying (\ref{b+ctoa})) , we have $b<\beta a$.
    \item \textbf{Prove} $\bm{a<\beta b}$.

    Suppose $c=kb$, where $\betaone<k<\beta$. Since $b+c>\beta d$, we have:
    \begin{equation}
    \label{dtob}
    d<\frac{1+k}{\beta}b
    \end{equation}

    From the above inequalities, we have:
    \begin{align*}
    a+b+c&=a+b+kb~~~~~~~~({\rm{applying (\ref{a+b+csmallerd}}}))\\
     &<(1+\beta+\betaone)d~~~({\rm{applying (\ref{dtob}}}))\\
     &<(1+\beta+\betaone)\times \frac{1+k}{\beta}b\\
    \Rightarrow a&< \left(\frac{1+\beta+\betaone}{\beta}-1\right)(1+k)b\\
    &=\frac{\beta+1}{\beta^2}(1+k)b\\
    &<\frac{(\beta+1)^2}{\beta^2}b
    \end{align*}
    When $\beta> \frac{7}{9}\times (\frac{\sqrt{837}+47}{54})^{-1/3}+(\frac{\sqrt{837}+47}{54})^{1/3} + \frac{1}{3} \approx 2.1479$, we have $\frac{(\beta+1)^2}{\beta}<\beta$. Hence $a<\beta b$.
    \end{enumerate}
    \item \textbf{Prove $\bm C$ is weight balanced to $\bm D$}.
    \begin{enumerate}[i.]
    \item \textbf{Prove} $\bm{d\le\beta c}$.

    When $\beta>\frac{1+\sqrt{5}}{2}\approx 1.618$, we have $\beta>1+\betaone$. Assume to the contrary $c<\betaone d$, we have $b<\beta c < d$. Thus:
    \begin{align*}
    b+c<(1+\betaone)d<\beta d
    \end{align*}
    , which contradicts with (\ref{b+clargerthand}) that $B+C$ is too heavy to be balanced with $D$.
    \item \textbf{Prove} $\bm{c<\beta d}$.

    Plug (\ref{alargerd}) in (\ref{a+b+csmallerd}) and get $b+c<(\beta+\betaone)d$. Recall that $\beta>1$, we have:
    \begin{align*}
    &\betaone c+c < b+c < (\beta+\betaone)d\\
    \Rightarrow& c<\frac{\beta^2+1}{\beta+1}d < \beta d
    \end{align*}
    \end{enumerate}
    \item \textbf{Prove $\bm{A+B}$ is weight balanced with $\bm{C+D}$}.
    \begin{enumerate}[i.]
    \item \textbf{Prove} $\bm{c+d\le\beta(a+b)}$.

    From (\ref{btoc}) we have $c<\beta b$, also $d<a<\beta a$ (applying (\ref{alargerd})), thus $c+d<\beta(a+b)$.
    \item \textbf{Prove} $\bm{a+b\le\beta(c+d)}$.

    Recall when $\beta>\frac{1+\sqrt{5}}{2}\approx 1.618$, we have $\beta>1+\betaone$. Applying (\ref{b+clargerthand}) and (\ref{btoc}) we have:
    \begin{align*}
    &d\le\betaone(b+c)\le c+\betaone c <\beta c\\
    \Rightarrow&\betaone d< c\\
    \Rightarrow&(1+\betaone)d<(1+\beta)c\\
    \Rightarrow&(1+\betaone+\beta)d<\beta(c+d)+c~~\apply{\ref{a+b+csmallerd}}\\
    \Rightarrow& a+b+c<(1+\betaone+\beta)d<\beta(c+d)+c\\
    \Rightarrow& a+b<\beta(c+d)
    \end{align*}
    \end{enumerate}
  \end{enumerate}
\end{enumerate}
Taking all the three conclusions into consideration, after either a single rotation or a double rotation, the new subtree will be rebalanced.

Then by induction we can prove Lemma \ref{lem:wbvalid}.
\end{proof} 

\end{document}